\documentclass[acmsmall]{acmart}

\usepackage[ruled,linesnumbered]{algorithm2e}
\usepackage{algpseudocode}
\usepackage{multirow}
\usepackage{enumitem}
\usepackage{diagbox}
\usepackage{array}
\usepackage{booktabs}
\usepackage{soul}
\usepackage{amsmath,bm}
\usepackage{microtype}
\usepackage{diagbox}
\usepackage{amsmath}
\usepackage{graphics}
\usepackage{enumitem} 
\usepackage{mathtools} 
\usepackage{setspace}
\usepackage{multirow}
\usepackage{epsfig}
\usepackage{xspace}
\usepackage{extarrows}
\usepackage{bm}
\usepackage{listings}
\usepackage{eucal} 
\usepackage{url}
\newcommand{\sem}[1]{\llbracket{#1}\rrbracket}

\newcommand{\code}[1]{\texttt{#1}}

\newcommand{\jrydel}[1]{}

\newcommand{\cat}{+ \!\!\! + \,}

 \newcommand{\mainname}{\textit{PolyGen}\xspace}
 \newcommand{\variable}[1]{\textit{#1}}

\newcommand{\rycomment}[1]{\rycomment{Ruyi}{#1}}

\acmJournal{PACMPL}

\setcopyright{none}

\usepackage{booktabs}   
\usepackage{subfigure}

\begin{document}
	
	\title[Occam Learning Meets Synthesis Through Unification]{Occam Learning Meets Synthesis Through Unification}      
	
	\author{Ruyi Ji}
	\affiliation{
	  \streetaddress{Key Lab of High Confidence Software Technologies, Ministry of Education Department of Computer Science and Technology, EECS}
	  \institution{Peking University}
	  \city{Beijing}
	  \country{China}            
	}
	\email{jiruyi910387714@pku.edu.cn}          
	
	\author{Jingtao Xia}
	\affiliation{
	  \streetaddress{Key Lab of High Confidence Software Technologies, Ministry of Education Department of Computer Science and Technology, EECS}
	  \institution{Peking University}
	  \city{Beijing}
	  \country{China}            
	}
	\email{xiajt@pku.edu.cn}          
	\author{Yingfei Xiong}
	\authornote{Corresponding author}
	\affiliation{
	  \streetaddress{Key Lab of High Confidence Software Technologies, Ministry of Education Department of Computer Science and Technology, EECS}
	  \institution{Peking University}
	  \city{Beijing}
	  \country{China}            
	}
	\email{xiongyf@pku.edu.cn}          
	\author{Zhenjiang Hu}
	\affiliation{
	  \streetaddress{Key Lab of High Confidence Software Technologies, Ministry of Education Department of Computer Science and Technology, EECS}
	  \institution{Peking University}
	  \city{Beijing}
	  \country{China}            
	}
	\email{huzj@pku.edu.cn}          

	\begin{abstract}
The generalizability of PBE solvers is the key to the empirical synthesis performance. Despite the importance of generalizability, related studies on PBE solvers are still limited. In theory, few existing solvers provide theoretical guarantees on generalizability, and in practice, there is a lack of PBE solvers with satisfactory generalizability on important domains such as conditional linear integer arithmetic (CLIA). In this paper, we adopt a concept from the computational learning theory, Occam learning, and perform a comprehensive study on the framework of synthesis through unification (STUN), a state-of-the-art framework for synthesizing programs with nested \texttt{if-then-else} operators. We prove that \textit{Eusolver}, a state-of-the-art STUN solver, does not satisfy the condition of Occam learning, and then we design a novel STUN solver, \textit{PolyGen}, of which the generalizability is theoretically guaranteed by Occam learning. We evaluate \textit{PolyGen} on the domains of CLIA and demonstrate that \textit{PolyGen} significantly outperforms two state-of-the-art PBE solvers on CLIA, \textit{Eusolver} and \textit{Euphony}, on both generalizability and efficiency.
	\end{abstract}
	

	\maketitle
\section{Introduction} \label{section:introduction}

In the past decades, oracle-guided inductive program synthesis (OGIS)~\cite{DBLP:journals/acta/JhaS17} receives much attention. In each iteration of OGIS, an oracle provides input-output examples to an inductive program synthesizer, or programming-by-example (PBE) synthesizer~\cite{DBLP:conf/ijcai/ShawWG75}, and the PBE synthesizer generates a program based on the examples. There are two typical types of OGIS problems. In the first type, the oracle can verify whether the synthesized program is correct, and provides a counter-example if the program is incorrect. Many applications under the counter-example guided inductive synthesis (CEGIS) framework~\cite{DBLP:conf/asplos/Solar-LezamaTBSS06} fall into this type. In the second type, the oracle cannot verify the correctness of the synthesized program but can provide a set of input-output examples. This includes the applications where the oracle is a black-box program, such as binary programs~\cite{zhai2016automatic}, and applications where the program is too complex to verify its correctness, e.g., the task involves system calls or complex loops, such as program repair, second-order execution, and deobfuscation~\cite{DirectFix,mechtaev2018symbolic,david2020qsynth,DBLP:conf/uss/BlazytkoCAH17,jha2010oracle}. 

In both types of problems, the generalizability of the PBE solver is the key to synthesis performance. In the first type, generalizability significantly affects the efficiency: the fewer examples the solver needs to synthesize a correct program, the fewer CEGIS iterations the synthesis requires, and thus the faster the synthesis would be. In the second type, the generalizability of the PBE solver decides the correctness of the whole OGIS system.

Despite the importance of generalizability, the studies on the generalizability of the existing PBE solvers are still limited. On the theoretical side, as far as we are aware, no existing PBE solver provides theoretical guarantees on generalizability. On the practical side, the generalizability of the existing PBE solvers is not satisfactory. As our evaluation will demonstrate later, on a synthesis task for solving the maximum segment sum problem, \textit{Eusolver}~\cite{DBLP:conf/tacas/AlurRU17}, a state-of-the-art PBE solver, uses $393$ examples to find the correct program, while our solver uses only $10$. 


In this paper, we propose a novel PBE solver, \mainname, that provides a theoretical guarantee on generalizability by construction. We adopt a concept from the computational learning theory, Occam learning~\cite{DBLP:journals/ipl/BlumerEHW87}, and prove that \mainname is an Occam solver, i.e., a PBE solver that satisfies the condition of Occam learning. 
A PBE solver is an Occam solver if, for any possible target program consistent with the given examples, the size of the synthesized program is guaranteed to be polynomial to the target program and sub-linear to the number of provided examples with a high probability. In other words, an Occam solver would prefer smaller programs to larger programs and thus follows the principle of Occam's Razor. 
In theory, \citet{DBLP:journals/ipl/BlumerEHW87} have proved that, given any expected accuracy, the number of examples needed by an Occam solver to guarantee the accuracy is bounded by a polynomial on the size of the target program. In practice, Occam learning has exhibited good generalizability in different domains~\cite{DBLP:conf/stoc/KearnsL88,DBLP:journals/jcss/KearnsS94,DBLP:journals/ml/AngluinL87,DBLP:conf/colt/Natarajan93,DBLP:journals/iandc/AldousV95}.

\mainname follows the \emph{synthesis through unification} (STUN)~\cite{DBLP:conf/cav/AlurCR15} framework. STUN is a framework for synthesizing nested \code{if-then-else} programs, and the solvers based on STUN such as \textit{Eusolver}~\cite{DBLP:conf/tacas/AlurRU17} and \textit{Euphony}~\cite{DBLP:conf/pldi/LeeHAN18} achieve the state-of-the-art results on many important benchmarks, e.g., the CLIA track in the SyGuS competition. A typical STUN solver consists of a term solver and a unifier. First the term solver synthesizes a set of \texttt{if}-terms, each being correct for a different subset of the input space, and then the unifier synthesizes \texttt{if}-conditions that combine the terms with \code{if-then-else} into a correct program for the whole input space. 

We first analyze a state-of-the-art STUN solver, \textit{Eusolver}~\cite{DBLP:conf/tacas/AlurRU17}, and prove that \textit{Eusolver} is not an Occam solver. Then we proceed to design \mainname. A key challenge to design an Occam solver is to scale up while satisfying the condition of Occam learning. For example, a trivial approach to ensuring Occam learning is to enumerate programs from small to large, and returns the first program consistent with the examples. However, this approach only scales to small programs. To ensure scalability, we divide the synthesis task into subtasks each synthesizing a subpart of the program, and propagate the condition of Occam learning into a sufficient set of conditions, where each condition is defined on a subtask. Roughly, these conditions require that each subtask synthesizes either a small program or a set of programs whose total size is small. 
Then, we find efficient synthesis algorithms that meet the respective conditions for each subtask. 

We instantiate \mainname on the domains of conditional linear integer arithmetic (CLIA), and evaluate \mainname against \textit{Esolver}~\cite{DBLP:conf/fmcad/AlurBJMRSSSTU13}, the best known PBE solver on CLIA that always synthesizes the smallest valid program, \textit{Eusolver}~\cite{DBLP:conf/tacas/AlurRU17} and \textit{Euphony}~\cite{DBLP:conf/pldi/LeeHAN18}, two state-of-the-art PBE solvers on CLIA. Our evaluation is conducted on $100$ benchmarks collected from the dataset of SyGuS-Comp~\cite{DBLP:journals/corr/abs-1904-07146} and an application of synthesizing divide-and-conquer algorithms~\cite{DBLP:conf/pldi/FarzanN17}. Besides, our evaluation considers two major oracle models in OGIS, corresponding to the applications where (1) the oracle can provide a counter-example for a given program, and (2) the oracle can only generate the correct output for a given input. Our evaluation results show that: 
\begin{itemize}
	\item Comparing with \textit{Esolver}, \mainname achieves almost the same generalizability while solving \textbf{9.55 times} more benchmarks than \textit{Esolver}.
	\item Comparing with \textit{Eusolver} and \textit{Euphony}, on efficiency, \mainname solves \textbf{43.08\%-79.63\%} more benchmarks with \textbf{$\times$7.02 - $\times$15.07 speed-ups}. Besides, on generalizability, \textit{Eusolver} and \textit{Euphony} requires \textbf{$\times$1.12 - $\times$2.42} examples comparing with \mainname on those jointly solved benchmarks. This ratio raises to \textbf{at least $\times$2.39 - $\times$3.33} when all benchmarks are considered.
\end{itemize}

To sum up, this paper makes the following contributions:
\begin{itemize}
	\item 
	We adopt the concept of Occam learning to the domain of PBE, prove that \textit{Eusolver} is not an Occam solver, and provide a sufficient set of conditions for individual components in the STUN framework to form an Occam solver (Section \ref{section:stun}).
	\item 
	We design a novel Occam solver based on the STUN framework, \mainname, by designing efficient algorithms for the two compoents that meet the above conditions. (Sections \ref{section:term-solver} and \ref{section:unifier}).
	\item 
	We instantiate \mainname to the domain of CLIA (Section \ref{section:implementation}) and evaluate $\mainname$ against state-of-the-art PBE solvers on CLIA (Section \ref{section:evaluation}). The evaluation results show that \mainname significantly outperforms \textit{Eusolver} and \textit{Euphony} on both efficiency and generalizability. 
\end{itemize}

\section{Related Work} \label{section:related}
\noindent \textbf{Generalizability of PBE Solvers.} Generalizability is known to be important for PBE solvers, and there have been different approaches proposed to improve generalizability. 

Guided by the principle of Occam's Razor, a major line of previous work converts the PBE task into an optimization problem by requiring the solver to find the simplest program~\cite{DBLP:conf/icml/LiangJK10,DBLP:conf/popl/Gulwani11, DBLP:conf/popl/RaychevBVK16, DBLP:conf/icse/MechtaevYR15}. This method has been evaluated to be effective in different domains, such as user-interacting PBE systems~\cite{DBLP:conf/popl/Gulwani11} and program repair~\cite{DBLP:conf/icse/MechtaevYR15}. However, the usage of this method is limited by efficiency, as in theory, requiring the optimality of the solution would greatly increase the difficulty of a problem. For many important domains, there is still a lack of an efficient enough PBE solver which implementing this method. For example, on the domains of CLIA, our evaluation shows that \textit{Eusolver}, a state-of-the-art PBE solver, solves $6.33$ times more benchmarks than \textit{Esolver}~\cite{DBLP:conf/fmcad/AlurBJMRSSSTU13}, the known best PBE solver on CLIA that guarantees to return the simplest program. 

Comparing with these previous work, though our paper is also based on the principle of Occam's Razor, we relax the constraint on the PBE solver by adopting the concept of Occam Learning~\cite{DBLP:journals/ipl/BlumerEHW87} from computational learning theory. Occam learning allows the solver to return a program that is at most polynomially worse than the optimal and still has theoretical guarantees on generalizability. While designing an Occam solver, we have more space to improve the efficiency than designing a solver optimizing the size. In this way, we successfully implement a PBE solver on CLIA that performs well on both efficiency and generalizability. 

Another line of work uses learned models to guide the synthesis procedure, and thus focuses on only probable programs~\cite{DBLP:conf/iclr/BalogGBNT17, DBLP:conf/iclr/KalyanMPBJG18, DBLP:journals/pacmpl/JiS0H20, DBLP:conf/cav/SinghG15, DBLP:conf/sigsoft/ChenMF19, DBLP:conf/icml/MenonTGLK13, DBLP:conf/icml/DevlinUBSMK17, DBLP:conf/pldi/LeeHAN18}. However, the efficiency of these approaches depends on domain knowledge. For example, \citet{DBLP:conf/iclr/KalyanMPBJG18} use input-output examples to predict the structure of the target program on the string manipulation domain: The effectiveness of their model relies on the structural information provided by strings and thus is unavailable on those unstructured domains, such as CLIA. In our evaluation, we evaluate a state-of-the-art PBE solver based on learned models, namely \textit{Euphony}~\cite{DBLP:conf/pldi/LeeHAN18}, and the result shows that its effectiveness is limited on CLIA.


\noindent \textbf{Analysis on the generalizability.} Analyzing the generalizability of learning algorithms is an important problem in machine learning and has been studied for decades. The \textit{probably approximately correct (PAC)} learnability  \cite{DBLP:journals/cacm/Valiant84} is a widely used framework for analyzing generalizability. When discussing PAC learnability of a learning task, the goal is to find a learning algorithm that (1) runs in polynomial time, (2) requires only a polynomial number of examples to achieve any requirement on the accuracy. On the synthesis side, there has been a line of previous work on the PAC learnability of logic programs~\cite{DBLP:journals/jair/Cohen95,DBLP:journals/jair/Cohen95a,DBLP:conf/colt/DzeroskiMR92}.  Besides, some approaches use the framework of PAC learnability to analyze the number of examples required by some specific algorithm~\cite{DBLP:journals/ml/LauWDW03, DBLP:conf/cav/DrewsAD19}.

In this paper, we seek a theoretical model that can compare the generalizability of different PBE solvers. At this time, the requirement on the generalizability provided by PAC learnability is too loose: According to the general bound provided by \citet{DBLP:journals/ipl/BlumerEHW87}, when the program space is finite, this condition is satisfied by any valid PBE solver. Therefore, we adopt another concept, Occam Learning, from computational learning theory. Comparing with PAC learnability, Occam learning (1) has a higher requirement on generalizability, as shown by~\citet{DBLP:journals/ipl/BlumerEHW87}, and (2) can reflect some empirical results in program synthesis, such as a PBE solver that always returns the simplest program should have better generalizability than an arbitrary PBE solver. To our knowledge, we are the first to introduce Occam Learning into program synthesis.

\noindent \textbf{Synthesizing CLIA Programs.} As mentioned in the introduction, our approach is implemented in the CLIA domain. CLIA is an important domain for program synthesis, as CLIA programs widely exist in real-world projects and can express complex behaviors by using nested \texttt{if-then-else} operators. There have been many applications of CLIA synthesizers,  such as program repair~\cite{DBLP:conf/icse/MechtaevYR15, DBLP:conf/sigsoft/LeCLGV17}, automated parallelization~\cite{DBLP:conf/pldi/FarzanN17, DBLP:conf/pldi/MoritaMMHT07}. On CLIA, PBE solvers are usually built on the STUN framework~\cite{DBLP:conf/cav/AlurCR15}, which firstly synthesizes a set of \texttt{if}-terms by a term solver, and then unifies them into a program by a unifier. There are two state-of-the-art PBE solvers: 
\begin{itemize}
	\item \textit{Eusolver}~\cite{DBLP:conf/tacas/AlurRU17}, which is comprised of an enumerative term solver and a unifier based on a decision-tree learning algorithm.
	\item \textit{Euphony}~\cite{DBLP:conf/pldi/LeeHAN18}, which uses structural probability to guide the synthesis of \textit{Eusolver}.
\end{itemize}
\mainname also follows the STUN framework. We evaluate \mainname against these two solvers in Section \ref{section:evaluation}. The result shows that \mainname outperforms them on both efficiency and generalizability.

Outside PBE, there are other techniques proposed for synthesizing CLIA programs:
\begin{itemize}
	\item \textit{DryadSynth}~\cite{DBLP:conf/pldi/HuangQSW20} reconciles enumerative and deductive synthesis techniques. As \textit{DryadSynth} requires a logic specificaion, it is not suitable for PBE tasks.
	\item \textit{CVC4}~\cite{DBLP:conf/cav/ReynoldsBNBT19} synthesizes programs from unsatisfiability proofs given by theory solvers. Though \textit{CVC4} is runnable on PBE tasks, it seldom generalizes from examples. We test \textit{CVC4} on a simple task where the target is to synthesize a program that returns the maximal value among three inputs. After requiring $300$ random examples, the error rate of the program synthesized by \textit{CVC4} on a random input is still larger than $97\%$. In contrast, \mainname requires only $12.2$ examples on average to synthesize a completely correct program.
\end{itemize}

There are also approaches on synthesizing boolean conditions, which is an important part in CLIA~\cite{DBLP:journals/tse/ErnstCGN01, DBLP:journals/corr/PadhiM17}. However, none of them discuss the theoretical guarantees on the generalizability, and it is unknown whether these approaches are Occam solvers.
\section{Motivating Example and Approach Overview} \label{section:example}
In this section, we introduce the basic idea of our approach via a motivating example adopted from benchmark \texttt{mpg\_ite2.sl} in the SyGuS competition. The target program $p^*$ is shown as the following, where $x, y, z$ are three integer inputs.
\begin{align*}
	p^*(x, y, z) \coloneqq & \texttt{ if }(x + y \geq 1)\texttt{ then }\{ \\
	& \ \ \ \ \texttt{ if }(x + z \geq 1)\texttt{ then }\{x + 1\}\texttt{ else }\{y + 1\} \\ 
	& \texttt{ }\}\texttt{ else }\{ \\
	& \ \ \ \ \texttt{ if }(y + z \geq 1)\texttt{ then }\{z + 1\}\texttt{ else }\{y + 1\} \\ 
	& \texttt{ }\}
\end{align*}   
We assume that there are $9$ input-output examples provided to the PBE solver. Table \ref{table:examples} lists these examples, where tuple $(x_0, y_0, z_0)$ in column $I$ represents an input where $x,y,z$ are set to $x_0,y_0,z_0$ respectively, and the \texttt{if}-term in column \textit{Term} represents the executed branch on each example.

\begin{table*}[!htbp]
	
	\caption{The input-output examples and the terms in the target program.} \label{table:examples}
	\begin{spacing}{1.1}
		\small
		\begin{tabular}{|c|c|c|c|c|c|c|c|c|c|c|c|}
			\hline
			ID & $I$ & $O$ & Term & ID & $I$ & $O$ & Term & ID & $I$ & $O$ & Term\\
			\hline
			$e_1$ & $(0, 1, 2)$ & $1$ & \multirow{3}{*}{$x + 1$} &$e_4$ & $(0, 2, 0)$ & $3$ & \multirow{3}{*}{$y + 1$} &$e_7$ & $(0, 0, 1)$ & $2$ & \multirow{3}{*}{$z + 1$} \\
			\cline{1-3} \cline{5-7} \cline{9-11}
			$e_2$ & $(1, 0, 2)$ & $2$ & & $e_5$ & $(-1, 3, 0)$ & $4$ & & $e_8$ & $(-3, 3, -2)$ & $-1$ & \\
			\cline{1-3} \cline{5-7} \cline{9-11}
			$e_3$ & $(-1, 3, 2)$ & $0$ & & $e_6$ & $(-1, 1, -1)$ & $2$ & & $e_9$ & $(-1, 0, 4)$ & $5$ & \\
			\hline
		\end{tabular}
	\end{spacing}
\end{table*}

\subsection{\textit{Eusolver}}
In this paper, we focus on designing an Occam solver for PBE tasks. As mentioned before, an Occam solver should synthesize programs whose size is polynomial to the size of the target program and sub-linear to the number of provided examples with high probability. Since the target program is unknown, an Occam solver should satisfy this requirement when the target program is the smallest valid program (programs consistent with the given input-output examples), and thus prefer smaller programs to larger programs.

We first show that \textit{Eusolver}~\cite{DBLP:conf/tacas/AlurRU17} may return unnecessarily large programs and thus is unlikely to be an Occam solver.  We shall formally prove that {\it Eusolver} is not an Occam solver in Section \ref{section:gen-eusolver}. \textit{Eusolver} follows the STUN framework and provides a term solver and an unifier. The term solver is responsible for synthesizing a term set that jointly covers all the given examples. In our example, one valid term set is $\{x + 1, y +1, z + 1\}$. A unifier is responsible for synthesizing a set of conditions to unify the term set into a program with nested \texttt{if-then-else} operators. In our example, the conditions are $\{x+y \geq 1, x+z \geq 1, y+z \geq 1\}$.

Similar to Occam learning, \textit{Eusolver} also tries to return small programs. To achieve this, \textit{Eusolver} enumerates the terms and the conditions from small to large, and then tries to combine the enumerated terms and conditions into a complete program. Though this approach controls the sizes of \textit{individual} terms and conditions, it fails to control the \textit{number} of the terms and the conditions, and thus may lead to unnecessarily large programs. 

The term solver of \textit{Eusolver} enumerates the terms from small to large, and includes a term in the term set when the set of examples covered by this term (i.e., the term is correct on these examples) is different from all smaller terms. The term set is complete when all examples are covered. As a result, this strategy may unnecessarily include many small terms, each covering only a few examples. In our motivating example, if constants such as $-1, \dots, 5$ are available, $\mathcal T_E$ will return set $\{-1, \dots, 5\}$ instead of $\{x+1, y+1, z+1\}$. Though such terms are small, the number of the terms grows with the number of examples.

The unifier of \textit{Eusolver} builds a decision tree using the ID3 algorithms. Here the terms are considered as labels, the enumerated \texttt{if-}conditions are considered as conditions, and a term is a valid label for an example if it covers the example. However, ID3 is designed for fuzzy classification problems, uses information gain to select conditions, and may select conditions that negatively contribute to program synthesis. For example, $x\ge 0$ will have good information gain for this example. In the original sets the three labels $x+1$, $y+1$ and $z+1$ are evenly distributed. Predicate $x\ge 0$ divides the examples into two sets where the distribution become uneven: in one set 50\% of the examples are labelled with $x+1$, and in the other set only 20\% of the examples are labelled $x+1$. However, in both sets the three labels still exist, and we still have to find conditions to distinguish them. As a result, selecting $x>0$ roughly doubles the size of the synthesized program. 


\subsection{\mainname}\begin{figure}[t]
	\centering  
	\vspace{-0.35cm} 
	\subfigtopskip=2pt 
	\subfigbottomskip=2pt\subfigure[The structure of \mainname.]{
		\raisebox{0.0\height}{
			\label{fig:frame}
			\includegraphics[width=0.4\linewidth]{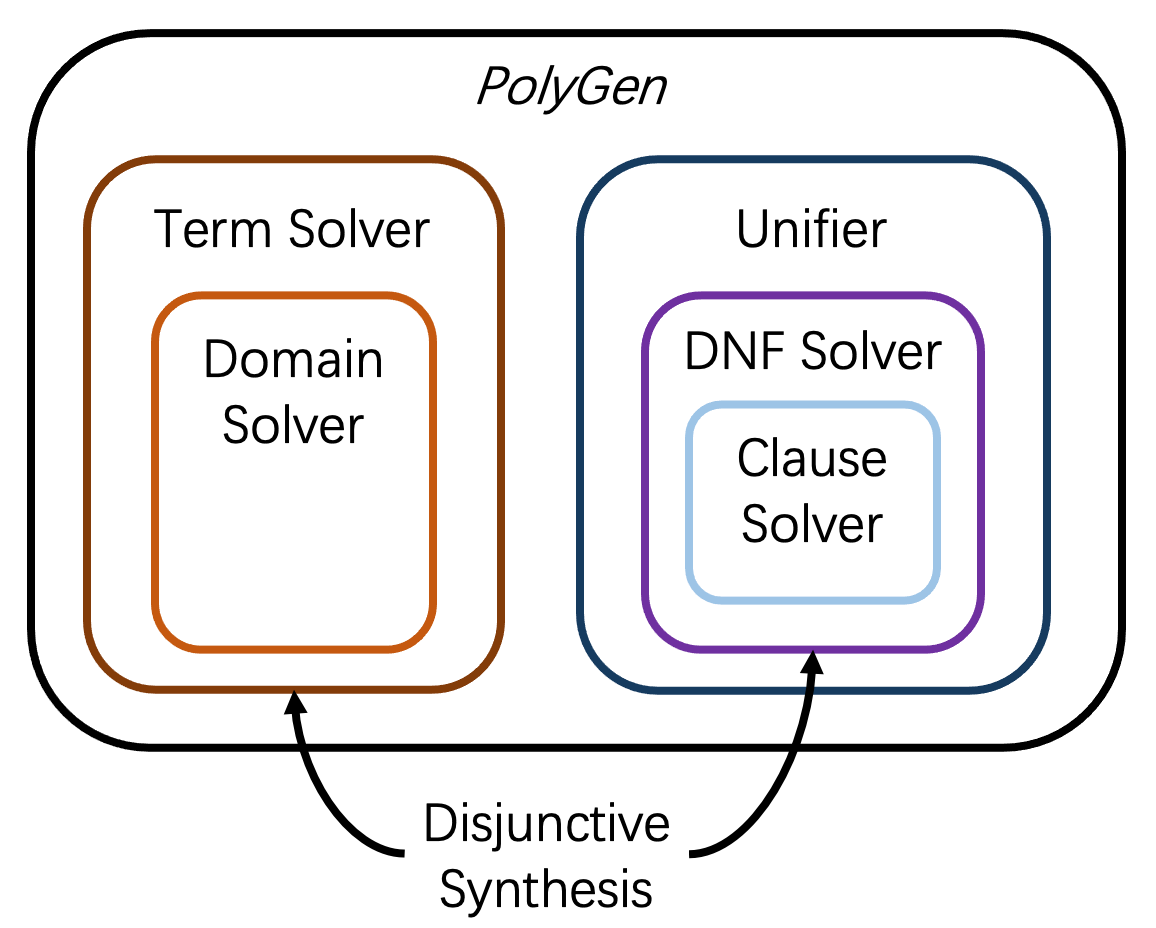}}
	
		}
	\subfigure[The program synthesized by $\mainname$ from $\{e_1, \dots, e_9\}$.]{
		\raisebox{0.7\height}{
		\label{fig:prog}
		\includegraphics[width=0.55\linewidth]{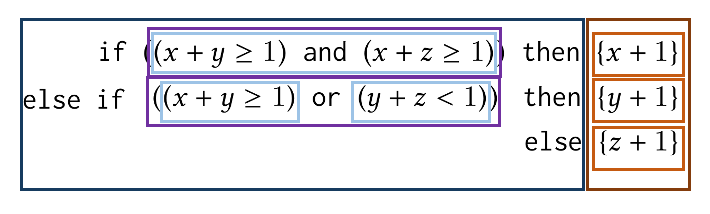}}}
	
	\caption{The left figure shows the structure of \mainname, where each sub-solver in \mainname is attached with a different color. The right figure shows the program synthesized by \mainname from examples $\{e_1, \dots, e_9\}$, where colored rectangles show the correspondence between partial programs and sub-solvers.}
\end{figure}

To synthesize small programs, \mainname controls not only the size of individual terms and conditions but also the number of conditions and terms. The structure of \mainname is shown as Figure~\ref{fig:frame}. As we can see in the figure, \mainname is built of a set of sub-solvers, each responsible for synthesizing part of the program. Figure~\ref{fig:prog} lists a program \mainname synthesizes, and the colored rectangles show the parts of the program synthesized by the sub-solver with the same color. In the following we shall illustrate how \mainname works.

\subsubsection{Term Solver}
To control the number of terms, the term solver of \mainname iterates a threshold on the number of terms, and tries to synthesize a term set whose size is equal to or smaller than this threshold. The threshold starts with a small number, and increases by a constant $c$ in each iteration. The process terminates if any iteration successfully synthesizes a term set. In each iteration, \mainname uses a randomized procedure to synthesize a term set. If a term set exists under a threshold, the probability that the term solver fails to synthesize a term set is bound by a constant $\epsilon$. 
Let us assume that a term set exists within the first iterated threshold. After $n$ iterations, the probability of failing to synthesize reduces to $\epsilon^n$, and the possible number of terms is only increased by $nc$. In other words, the number of the synthesized terms is guaranteed to be small with high probability. 

Now we explain how we implement the randomized procedure to obtain a term set with a bounded failure probability. 
We assume there is a domain solver that synthesizes a term based on a set of examples, and the domain solver is also an Occam solver. 
For illustration, let us assume currently the threshold for the number of terms is 3, and in our example there exists at least one term set $T=\{x + 1, y +1, z + 1\}$ under this threshold. 
The term solver first samples many subsets of the examples and invoke the domain solver to synthesize a term for each subset. If any subset is covered by a term in $T$, the domain solver will have a chance to synthesize the term. For example, if $e_1$ and $e_2$ are sampled, the domain solver has a chance to synthesize $x+1$ because of the generalizability of the Occam solver. 
As a result, if we sample enough subsets, we can synthesize a term in $T$ with any small bounded failure rate. Then for any successfully synthesized subsets, we repeat this procedure to synthesize terms for the remaining examples. For example, when $x+1$ is synthesized, the procedure continues with examples $e_4\ldots e_9$. The procedure ends when no example remains. Since in each turn the probability of failing to find a term in $T$ is bounded, the total probability of failing to find the term set $T$ is bounded. 
Please note the sizes of synthesized terms are guaranteed to be small as the domain solver is an Occam solver.

In the domain of CLIA, the \texttt{if}-terms are all linear integer expressions, and the domain solver can be implemented by finding the simplest valid term via linear integer programming, as we shall show in Section \ref{section:implementation}. 

\subsubsection{Unifier}
To control the number of conditions, instead of synthesizing a decision tree, the unifier of \mainname synthesizes \textit{decision list}~\cite{DBLP:journals/ml/Rivest87}. In a decision list, each condition distinguishes one term from the rest of the terms. Figure~\ref{fig:prog} shows the program synthesized by \mainname on examples $\{e_1, \dots, e_9\}$, which is semantically equivalent to the target program $p^*$. The number of conditions is equal to the number of terms minus one, and thus is bounded.

However, the conditions in a decision list may become larger and thus cannot be synthesized using an enumerative algorithm. To efficiently synthesize small conditions to distinguish the terms, we notice that the conditions are in the disjunctive normal forms (DNF) in the initial definition of decision lists, where a DNF is the disjunction of clauses, a clause is the conjunction of literals, and a literal is a predicate in the grammar or its negation. Then we design three sub-solvers for different parts of the conditions. The clause solver synthesizes clauses from the literals, where the literals are enumerated by size in the same way as {\it Eusolver}. The DNF solver synthesizes a DNF formula based on the clause solver. Finally, the unifier synthesizes all the conditions based on the DNF solver.

Given a set of terms, the unifier create a synthesis subtask for each term $t$, where the synthesized program has to return {\tt true} on example inputs covered by $t$ (positive examples) and return {\tt false} on the remaining example inputs (negative examples). For example, when synthesizing the condition for $y+1$, $e_4, e_5, e_6$ are positive examples and $e_7, e_8, e_9$ are negative examples. Here $e_1, e_2, e_3$ are already covered by $x+1$. 
Then the unifier invokes the DNF solver to solve these tasks. The conditions synthesized are guaranteed to be small if the DNF solver guarantees to return small conditions. 

Before getting into the DNF solver, let us discuss the clause solver first. Given a set of literals, a set of input-output examples where the output is Boolean, the clause solver returns a clause, i.e., the conjunction of a subset of literals satisfying all examples. The clause solver reduces this problem into \textit{weighted set covering}, and uses a standard approximation algorithm~\cite{DBLP:journals/mor/Chvatal79} to solve it. As will be formally proved later, the clause solver is an Occam solver. 

Based on the clause solver, we build the DNF solver. The DNF solver synthesizes a set of clauses, where all clauses should return \texttt{false} for each negative example, and at least one clause should return \texttt{true} for each positive example. 
We notice this synthesis problem has the same form as the term solver: given a set of examples (in this case, a set of positive examples) and an Occam solver (in this case, the clause solver), we need to synthesize a set of programs (in this case, a set of clauses) to cover these examples. Therefore, the DNF solver uses the same algorithm as the term solver, and we uniformly refer this algorithm as \emph{disjunctive synthesis}. Since the disjunctive synthesis algorithm guarantees the returned program set is small in terms of both the sizes of individual programs and the number of total programs, the DNF solver guarantees to return a condition of small size.

\section{Occam Learning} \label{section:background}

\subsection{Preliminaries: Programming by Example}\label{section:pre-pbe}
The problem of \textit{programming-by-example (PBE)} is usually discussed above an underlying domain $\mathbb D = (\mathbb P, \mathbb I, \mathbb O, \sem{.}_{\mathbb D})$, where $\mathbb P, \mathbb I, \mathbb O$ represent a program space, an input space, and an output space respectively, $\sem{.}_{\mathbb D}$ is an oracle function that associates a program $p \in \mathbb P$ and an input $I \in \mathbb I$ with an output in $\mathbb O$, denoted as $\sem{p}_{\mathbb D}(I)$. The domain limits the ranges of possibly used programs and concerned inputs and provides the semantics of these programs.

For simplicity, we make two assumptions on the domain: (1) There is a universal oracle function $\sem{.}$ for any domains; (2) The output space $\mathbb O$ is always induced by the program space $\mathbb P$, the input space $\mathbb I$, and the oracle function $\sem{.}$, i.e., $\mathbb O = \{\sem{p}(I)\ |\ p \in \mathbb P, I \in \mathbb I\}$.
In the remainder of this paper, we abbreviate a domain as a pair $(\mathbb P, \mathbb I)$ of a program space and an input space. We shall use notation $\mathcal F$ to represent a family of domains, and thus discuss general properties of domains.

PBE~\cite{DBLP:conf/ijcai/ShawWG75} is a subproblem of program synthesis where the solver is required to learn a program from a set of given input-output examples. As this paper focuses  on the generalizability of PBE solvers, we assume that there is at least a program satisfying all given examples: The problem of determining whether there is a valid program is another domain of program synthesis, namely \textit{unrealizability}~\cite{DBLP:conf/pldi/HuCDR20,DBLP:journals/pacmpl/KimHDR21}, and is out of the scope of our paper.

\begin{definition}[Programming by Example] \label{def:pbe}Given a domain $\mathbb D$, a PBE task $T \in (\mathbb I \times \mathbb O)^*$ is a sequence of input-output examples. Define $\mathbb T(\mathbb D) \subseteq (\mathbb I \times  \mathbb O)^*$ as the set of PBE tasks where there is at least one program satisfying  all examples. PBE solver $\mathcal S$ is a function that takes a PBE task as the input, and returns a program satisfying all given examples, i.e., $\forall T \in \mathbb T(\mathbb D), \forall (I,O) \in T, \sem{\mathcal S(T)}(I) = O$.
\end{definition}

\subsection{Occam Learning and Occam Solver}
In computational learning theory, \textit{Occam Learning}~\cite{DBLP:journals/ipl/BlumerEHW87} is proposed to explain the effectiveness of the principle of Occam's Razor. In PBE, an Occam solver guarantees that the size of the synthesized program is at most polynomially larger than the size of the target program.
\begin{definition}[Occam Solver\footnote{The original definition of Occam solvers is only for deterministic algorithms. Here we extend its definition to random algorithms. We compare these two definitions in Appendix \ref{appendix:ocaam learning}.}]\label{definition:occam}For constants $\alpha \geq 1, 0 \leq \beta < 1$, PBE solver $\mathcal S$ is an $(\alpha, \beta)$-Occam solver on a family of domains $\mathcal F$ if there exist constants $c, \gamma > 0$ such that for any program domain $\mathbb D \in \mathcal F$, for any target program $p^* \in \mathbb P$, for any input set $\{I_1, \dots, I_n\} \subseteq \mathbb I$, for any error rate $\epsilon \in \left(0, \frac{1}{2}\right)$:
$$
\Pr\left[\text{size}\left(\mathcal S\big(T(p^*, I_1, \dots, I_n)\big)\right) > c\left(\text{size}(p^*)\right)^{\alpha}n^{\beta}\ln^{\gamma}\left(\frac{1}{\epsilon}\right)\right] \leq \epsilon
$$
where \text{size}(p) is the length of the binary representation of program $p$, $T(p^*, I_1, \dots, I_n)$ is defined as the PBE task corresponding to target program $p^*$ and inputs $I_1, \dots, I_n$.
\end{definition}

We assume the program domain is defined by a context-free grammar. At this time, a program can be represented by its left-most derivation and can be encoded as a sequence of grammar rules.

\begin{definition} \label{def:size} The size $\text{size}(p)$ of program $p$ is defined as $\lceil\log_2 N\rceil \times |p|$, where $|p|$ is the number of grammar rules used to derive $p$, and $N$ is the number of different grammar rules.
\end{definition}

\begin{example} When only input variable $x$, operator $+$ and constants $1, 2$ are available in the grammar, $\text{size}(x+1)$ is defined as $\lceil\log_2 4\rceil \times 3 = 6$. When there are $a$ input variables, $b$ constants and $c$ different operators available, $\text{size}(x + 1)$ is defined as $3 \lceil\log_2 (a + b+ c)\rceil$.
\end{example}

The size provides a logarithmic upper bound on the number of programs no larger than $p$.

\begin{lemma} \label{lemma:size}For any domain $\mathbb D$, $\forall p \in \mathbb P$, 
	$ \left|\left\{p' \in \mathbb P\  |\ \text{size}(p') \leq \text{size}(p)\right\} \right| \leq 2^{\text{size}(p)}
	$.
\end{lemma} 

\citet{DBLP:journals/ipl/BlumerEHW87} analysizes Occam solvers under the \textit{probably approximately correct (PAC)} learnability framework, and proves that the generalizability of Occam solvers is always guaranteed.

\begin{theorem} \label{theorem:occam-error}Let $\mathcal S$ be an $(\alpha, \beta)$-Occam solver on domain $\mathbb D$. Then there exist constants $c, \gamma >0$ such that for any $0 < \epsilon, \delta < 1$, for any distribution $D$ over $\mathbb I$ and any target program $p^* \in \mathbb P$:
	$$
	\small
	\forall n > c \left(\frac{1}{\epsilon}\ln\left(\frac{2}{\delta}\right) + \left(\frac{(\text{size}(p^*))^{\alpha}\ln^{\gamma}(2 / \delta)}{\epsilon}\right)^{1/(1-\beta)}\right), \Pr_{I_i \sim D} \left[\text{err}_{D,p^*}\left(\mathcal S\big(T(p^*, I_1, \dots, I_n)\big)\right) \geq \epsilon\right] \leq \delta
	$$
where $\text{err}_{D,p^*}(p)$ represents the error rate of program $p$ when the input distribution is $D$ and the target program is $p^*$, i.e., $\text{err}_{D,p^*}(p) \coloneqq \Pr_{I\sim D} \left[\sem{p}(I) \neq \sem{p^*}(I)\right]$. 
\end{theorem}

Due to space limit, we move all proofs to Appendix \ref{appendix:proofs}.

When $\epsilon$ and $\delta$ are fixed, Theorem \ref{theorem:occam-error} implies that an $(\alpha, \beta)-$Occam solver can find a program similar to the target $p^*$ with only $O\left(\text{size}(p^*)^{\alpha/(1-\beta)}\right)$ examples. Such a bound matches the principle of Occam's Razor, as it increases monotonically when the size of the target program increases.

The class of Occam solvers can reflect the practical generalizability of PBE solvers. Let us take two primitive solvers $\mathcal S_{\text{min}}$ and $\mathcal S_{\text{rand}}$ as an example. For any PBE task $T$, let $\mathbb P(T) \subseteq \mathbb P$ be the set of programs that are consistent with examples in $T$:
\begin{itemize}
	\item $\mathcal S_{\text{rand}}$ is the most trivial synthesizer that has no guarantee on the quality of the result: It just uniformly returns a program from $\mathbb P(T)$:
	$\forall p \in \mathbb P(T), \Pr\left[\mathcal S_{\text{rand}}(T) = p\right] = |\mathbb P(T)|^{-1}$.
	\item $\mathcal S_{\text{min}}$ regards a PBE task as an optimization problem, and always returns the syntactically smallest program in $\mathbb P(T)$:
	$\mathcal S_{\text{min}}(T) \coloneqq \arg \min_{p \in \mathbb P(T)} \text{size}(p)$.
\end{itemize}

In practice, it is usually believed that $\mathcal S_{\text{min}}$ has better generalizability than $\mathcal S_{\text{rand}}$. We prove that the class of Occam solvers can discover this advantage, as $\mathcal S_{\text{min}}$ is an Occam solver but $\mathcal S_{\text{rand}}$ is not.

\begin{theorem} \label{theorem:min-rand}Let $\mathcal F^A$ be the family of all possible domains. Then $\mathcal S_{\text{min}}$ is an $(1, 0)$-Occam solver on $\mathcal F^A$, and $\mathcal S_{\text{rand}}$ is not an Occam solver on $\mathcal F^A$.
\end{theorem}
\section{Synthesis Through Unification} \label{section:stun}

\subsection{Preliminaries: Synthesis through Unification}
The framework of \textit{synthesis through unification (STUN)} focuses on synthesis tasks for programs with nested \texttt{if-then-else} operators. Formally, STUN assumes that the program space can be decomposed into two subspaces, for \texttt{if}-conditions and \texttt{if}-terms respectively. 

\begin{definition}A program space $\mathbb P$ is a conditional program space if and only if there exist two program spaces $\mathbb P_c$ and $\mathbb P_t$ such that $\mathbb P$ is the smallest set of programs such that:
	$$
	\mathbb P = \mathbb P_t \cup \left\{\texttt{if}\ c\ \texttt{then}\ p_1\ \texttt{else}\ p_2\ |\ p_1, p_2 \in \mathbb P, c \in \mathbb P_c  \right\}
	$$
\end{definition}
We use pair $(\mathbb P_t, \mathbb P_c)$ to denote a conditional program space derived from term space $\mathbb P_t$ and condition space $\mathbb P_c$. Besides, we call a domain $\mathbb D$ \textit{conditional} if the program space in $\mathbb D$ is conditional.

A STUN solver synthesizes programs in two steps:
\begin{enumerate}
	\item A term solver is invoked to synthesize a set of programs $P \subseteq \mathbb P_t$ such that for any example, there is always a consistent program in $P$.
	\item A unifier is invoked to synthesize a valid program from conditional program space $(P, \mathbb P_c)$. 
\end{enumerate}
In this paper, we only consider the specialized version of the STUN framework on PBE tasks.

\begin{definition}[Term Solver] Given conditional domain $\mathbb D$, term solver $\mathcal T: \mathbb T(\mathbb D) \rightarrow \mathcal P(\mathbb P_t)$ returns a set of terms covering all examples in a given PBE task, where $\mathcal P(\mathbb P_t)$ denotes the power set of $\mathbb P_t$:
	$$\forall T \in \mathbb T(\mathbb D), \forall (I,O) \in T, \exists p \in \mathcal T(T), \sem{p}(I) = O$$
\end{definition}

\begin{definition}[Unifier] Given a conditional domain $\mathbb D$, a unifier $\mathcal U$ is a function such that for any set of terms $P \subseteq \mathbb P_t$, $\mathcal U(P)$ is a valid PBE solver for $(P, \mathbb P_c)$.
\end{definition}

A STUN solver consists of a term solver $\mathcal T$ and a unifier $\mathcal U$. Given a PBE task $T$, the solver returns $\mathcal U(\mathcal T(T))(T)$ as the synthesis result. \citet{DBLP:conf/cav/AlurCR15} proves that such a combination is complete when the conditional domain is \textit{if-closed}. For other domains, STUN can be extended to be complete by backtracking to the term solver when the unifier fails~\cite{DBLP:conf/cav/AlurCR15, DBLP:conf/tacas/AlurRU17}.

\begin{definition}[If-Closed] \label{def:if-closed}A conditional domain $\mathbb D$ is \textit{if-closed} if: $$\forall p_1, p_2 \in \mathbb P_t, \exists c \in \mathbb P_c, \forall I \in \mathbb I, \left(\sem{c}(I) \iff \sem{p_1}(I) = \sem{p_2}(I)\right)$$
\end{definition}
Please note that any conditional domain with equality is {\it if-closed}, as $c$ can be constructed by testing the equality between the outputs of $p_1$ and $p_2$. In the rest of the paper, we 
 assume the conditional program space is if-closed, and use $\mathcal F_{C}$ to denote a family of if-closed domains.

\textit{Eusolver}~\cite{DBLP:conf/tacas/AlurRU17} is a state-of-the-art solver following the STUN framework. It takes efficiency and generalizability as its design purposes, and makes a trade-off between them.

The term solver $\mathcal T_{E}$ in \textit{Eusolver} is motivated by $\mathcal S_{\text{min}}$. $\mathcal T_{E}$ enumerates terms in $\mathbb P_t$ in the increasing order of the size. For each term $t$, if there is no smaller term that performs the same with $p$ on examples, $t$ will be included in the result. $\mathcal T_E$ returns when the result is enough to cover all examples.

The unifier  $\mathcal U_{E}$ in \textit{Eusolver} regards nested \texttt{if-then-else} operators as a decision tree, and uses \textit{ID3}~\cite{DBLP:journals/ml/Quinlan86}, a standard decision-tree learning algorithm, to unify the terms. $\mathcal U_{E}$ learns a decision tree recursively: In each recursion, it first tries to use a term to cover all remaining examples. If there is no such term, $\mathcal U_{E}$ will heuristically pick up a condition $c$ from $\mathbb P_c$ as the \texttt{if}-condition. According to the semantics of $c$, the examples will be divided into two parts, which will be used to synthesize the \texttt{then}-branch and the \texttt{else}-branch respectively.

\subsection{Generalizability of STUN} \label{section:stun-gen}
In this section, we study the generalizability of the STUN framework. To start, we extend the concept of Occam solvers to term solvers and unifiers: For an Occam term solver, there should be a polynomial bound on the total size of synthesized terms, and for an Occam unifier, the induced PBE solver should always be an Occam solver for any possible term set. These definitions will be used to guide our design of \mainname later.

\begin{definition} \label{def:size-preserve} For constants $\alpha \geq 1, 0 \leq \beta < 1$, term solver $\mathcal T$ is an $(\alpha, \beta)$-Occam term solver on $\mathcal F_C$ if there exist constants $c, \gamma > 0$ such that for any domain $\mathbb D \in \mathcal F_C$, for any target program $p^* \in \mathbb P$, for any input set $\{I_1, \dots, I_n\} \subseteq \mathbb I$, for any error rate $\epsilon \in \left(0, \frac{1}{2}\right)$: 
$$
\Pr\left[\text{tsize}\left(\mathcal T\big(T(p^*, I_1, \dots, I_n)\big)\right) > c\left(\text{size}(p^*)\right)^{\alpha}n^{\beta}\ln^{\gamma}\left(\frac{1}{\epsilon}\right)\right] \leq \epsilon
$$
where $\text{tsize}(P)$ is the total size of terms in term set $P$, i.e., $\sum_{t \in P} \text{size}(t)$.
\end{definition}

\begin{definition}\label{def:poly-unifier} For constants $\alpha \geq 1, 0 \leq \beta < 1$, unfier solver $\mathcal U$ is an $(\alpha, \beta)$-Occam unifier on $\mathcal F_C$ if there exist constants $c, \gamma > 0$ such that for any domain $\mathbb D \in \mathcal F_C$, for any term set $P \subseteq \mathbb P_t$, for any target program $p^* \in (P, \mathbb P_t)$, for any input set $\vec{I} = \{I_1, \dots, I_n\} \subseteq \mathbb I$, for any error rate $\epsilon \in \left(0, \frac{1}{2}\right)$: 
	\begin{align*}
	\Pr\left[\text{size}\left(\mathcal U(P)\big(T(p^*, I_1, \dots, I_n)\big)\right) > c\big(\max(\text{size}(p^*), \text{tsize}(P))\big)^{\alpha}n^{\beta}\ln^{\gamma}\left(\frac{1}{\epsilon}\right)\right] \leq \epsilon
	\end{align*}
\end{definition}

In Definition \ref{def:poly-unifier}, besides the size of the target program, the bound also refers to the total size of $P$. Such relaxation allows the unifier to use more examples when a large term set is provided.

Based on  the above definitions, we prove that under some conditions, an STUN solver comprised of an Occam term solver and an Occam unifier is also an Occam solver.

\begin{theorem} \label{theorem:gen-stun} Let $\mathcal F_C$ be a family of if-closed conditional domains, $\mathcal T$ be an $(\alpha_1, \beta_1)$-Occam term solver on $\mathcal F_C$, $\mathcal U$ be an $(\alpha_2, \beta_2)$-Occam unifier where $\beta_1\alpha_2 + \beta_2< 1$. Then the STUN solver comprised of $\mathcal T$ and $\mathcal U$ is an $((\alpha_1 + 1)\alpha_2, \beta_1\alpha_2 +\beta_2)$-Occam solver.
\end{theorem}

\subsection{Generalizability of \textit{Eusolver}} \label{section:gen-eusolver}
In this section, we analyze the generalizability of \textit{Eusolver} and prove that \textit{Eusolver} is not an Occam solver. We start from the term solver $\mathcal T_{E}$. As $\mathcal T_{E}$ enumerates terms in the increasing order of the size, $\mathcal T_{E}$ guarantees that all synthesized terms are small. However, the main problem of $\mathcal T_{E}$ is that it does not control the total number of synthesized terms. Therefore, the total size of the term set returned by $\mathcal T_{E}$ can be extremely large, as shown in Example \ref{example:eusolver-size}.

\begin{example} \label{example:eusolver-size} Consider the following term space $\mathbb P_t^n$, input space $\mathbb I_t^n$ and target program $p$:
	$$
	\mathbb P_t^n = \{2, 3, \dots, n + 1, x + 1\}\ \ \ \mathbb I_t^n = [1, n] \cap \mathbb Z\ \ \ p = x+ 1
	$$
	
	As $p$ is the largest term in $\mathbb P_t^n$, on any input $x_0$ in $\mathbb I_t^n$, there is always a smaller term $c$ that performs the same as $p$, where $c$ is a constant equal to $x_0 + 1$. Therefore, whatever the PBE task is, $\mathcal T_{E}$ always returns a subset of $P_A = \{2,3, \dots, n + 1\}$ and never enumerates to the target program $p$. 
	
	When all inputs in $\mathbb I_t^n$ are included in the PBE task, the term set synthesized by $\mathcal T_{E}$ is always $P_A$. At this time, the total size of $P_A$ is $n\lceil\log_2(n +3) \rceil$, the number of examples is $n$ and the size of the target program is $\lceil \log_2(n + 3) \rceil$. Clearly, there are no $\alpha \geq 1, 0 < \beta < 1$ and $c > 0$ such that $\forall n, \text{tsize}(P_A) \leq c\left(\text{size}(p)\right)^{\alpha}n^{\beta}$. Therefore, $\mathcal T_{E}$ is not an Occam term solver.
	
	Moreover, at this time, the program synthesized by \textit{Eusolver} must utilize all terms in $P_A$, and thus its size is no smaller than $\text{tsize}(P_A)$. So $\textit{Eusolver}$ is not an Occam solver as well.
\end{example}

The following fact comes from Example \ref{example:eusolver-size} immediately.

\begin{theorem}\label{theorem:eusolver-term} $\mathcal T_E$ is not an Occam term solver on $\mathcal F^A_C$, and \textit{Eusolver} is not an Occam solver on $\mathcal F^A_C$, where $\mathcal F^A_C$ is the family of all if-closed conditional domains.
\end{theorem}

The generalizability of $\mathcal U_{E}$ is related to the underlying decision-tree learning algorithm $\textit{ID3}$. \citet{DBLP:conf/stacs/HancockJLT95} proves that there is no polynomial-time algorithm for learning decision trees that generalizes within a polynomial number of examples unless $\mathsf{NP} = \mathsf{RP}$, where $\mathsf{RP}$ represents the class of polynomial-time random algorithms with one-side error. Combining with Theorem \ref{theorem:occam-error}, we obtain the following lemma, which implies that $\mathcal U_{E}$ is unlikely to be an Occam unifier.

\begin{theorem}\label{theorem:eusovler-unifier} There is no polynomial-time Occam unifier on $\mathcal F^A_C$ unless $\mathsf{NP} = \mathsf{RP}$.
\end{theorem}

These results indicate that (1) {\it Eusolver} itself is not an Occam solver, and (2) if we would like to design an Occam solver following Theorem~\ref{theorem:gen-stun}, neither $\mathcal T_E$ nor $\mathcal U_E$ can be reused.

\section{Term Solver} \label{section:term-solver}

\subsection{Overview} \label{section:term-main}
For an Occam term solver, the total size of returned terms should be bounded. Therefore, a term solver must be an Occam solver if (1) the number of returned terms is bounded, and (2) the maximal size of returned terms is bounded. 

\begin{lemma} \label{lemma:term-split} For constants $\alpha_1,\alpha_2 \geq 0, 0 \leq \beta_1, \beta_2 < 1$ where $\beta_1 + \beta_2 < 1$, term solver $\mathcal T$ is an $(\alpha_1 + \alpha_2, \beta_1 + \beta_2)$-Occam solver on $\mathcal F_C$ if there exist constants $c, \gamma > 0$ such that for any conditional domain $\mathbb D \in \mathcal F_C$, any target program $p^* \in \mathbb P$, and any input set $\{I_1, \dots, I_n\} \subseteq \mathbb I$:
	\begin{enumerate}
		\item With a high probability, the size of terms returned by $\mathcal T$ is bounded by $\text{size}(p^*)^{\alpha_1} n^{\beta_1}$.
		$$
		\Pr\left[\max \left\{ \text{size(p)}\ \big |\ p \in \mathcal T \big((I_1, \sem{p}(I_1)), \dots, (I_n, \sem{p}(I_n))\big)\right\}  > c\left(\text{size}(p^*)\right)^{\alpha_1}n^{\beta_1}\ln^{\gamma}\left(\frac{1}{\epsilon}\right)\right] \leq \epsilon
		$$
		\item With a high probability, the number of terms returned by $\mathcal T$ is bounded by $\text{size}(p^*)^{\alpha_2} n^{\beta_2}$.
		$$
		\Pr\left[\left|\mathcal T \big((I_1, \sem{p}(I_1)), \dots, (I_n, \sem{p}(I_n))\big)\right|  > c\left(\text{size}(p^*)\right)^{\alpha_2}n^{\beta_2}\ln^{\gamma}\left(\frac{1}{\epsilon}\right)\right] \leq \epsilon
		$$
	\end{enumerate}
\end{lemma}

In Lemma \ref{lemma:term-split}, the first condition has a form similar to the guarantee provided by an Occam solver. Motivated by this point, we design $\mathcal T_{\text{poly}}$ as a meta-solver that takes an Occam solver $\mathcal S_t$ on the term space as an input. To solve a term finding task, $\mathcal T_{\text{poly}}$ firstly decomposes it into several standard PBE tasks and then invokes $\mathcal S_t$ to synthesize terms with bounded sizes.

One challenge is that, in a term finding task, different examples correspond to different target terms, i.e., \texttt{if}-terms used in the target program. To find a target term using $\mathcal S_t$, $\mathcal T$ should pick up enough examples that correspond to the same target term. To do so, we utilize the fact that there must be a target term that covers a considerable portion of all examples, as shown in Lemma \ref{lemma:sample}.

\begin{lemma} \label{lemma:sample}Let $T$ be a PBE task, and let $P$ be a set of terms that covers all examples in $T$, i.e., $\forall (I, O) \in T, \exists p \in P, \left(\sem{p}(I) = O\right)$. There is always a term $p \in P$ such that:
	$$
	\left|\big\{(I, O) \in T\ \big|\ \sem{p}(I) = O\big\} \right| \geq |T|/|P|
	$$
\end{lemma}
Given a term finding task $T$ where $P^*$ is the set of target terms, let $t^* \in P^*$ be the term that covers the most examples. According to Lemma \ref{lemma:sample}, if we randomly select $n_t$ examples from $T$, term $t^*$ will be consistent with all selected examples with a probability of at least $|P^*|^{-n_t}$. Therefore, $\mathcal T_{\text{poly}}$ repeatedly invokes $\mathcal S_{t}$ on a small set of random examples drawn from $T$: When the generalizability of $\mathcal S_t$ is guaranteed on term space $\mathbb P_t$, the number of examples and the number of turns are large enough, $\mathcal T_{\text{poly}}$ would find a term semantically similar with $t^*$ with a high probability.

For the second condition, $\mathcal T_{\text{poly}}$ assumes that there is an upper bound $k$, and searches among term sets with at most $k$ terms. If the search process guarantees to find a valid term set with a high probability when $k$ is larger than a small bound, which is polynomial to the size of the target program and sub-linear to the number of examples, the second condition of Lemma \ref{lemma:sample} can be satisfied by iteratively trying all possible $k$.

\subsection{Algorithm} \label{section:term-algorithm}
The pseudo-code of $\mathcal T_{\text{poly}}$ is shown as Algorithm \ref{alg:term-main}. $\mathcal T_{\text{poly}}$ is configured by a domain solver $\mathcal S_t$, which is used to synthesize terms, and a constant $c$, which is used to configure bounds used in $\mathcal T_{\text{poly}}$. We assume that $\mathcal S_t$ can discover the case where there is no valid solution, and returns $\bot$ at this time.
\SetKwFunction{Search}{Search}
\SetKwFunction{Get}{GetCandidatePrograms}
\SetKwFunction{G}{Get}
\SetKwFunction{Covered}{Covered}
\begin{algorithm}[t]
	\small
	\caption{The term solver $\mathcal T_{\text{poly}}$ in \mainname.}
	\label{alg:term-main}
	\LinesNumbered
	\KwIn{A PBE task $T$, an $(\alpha, \beta)$-Occam sovler $\mathcal S_t$ on the term space $\mathbb P_t$ and a constant $c$.}
	\KwOut{A set of programs $P$ that covers all examples.}
	\SetKwProg{Fn}{Function}{:}{}
	\Fn{\Get{\variable{examples}, $k$, $n_t$, $s$}}{
		$\variable{result} \gets \{\}$; \\
		\For{each $\variable{turn} \in [1, n_t \times k^{n_t}]$}{
			Uniformly and independently samples $n_t$ examples $e_1, \dots, e_{n_t}$ from $\variable{examples}$; \\
			$p \gets \mathcal S_t\left(e_1, \dots, e_{n_t}\right)$; \\
			\If {$p \neq \bot \wedge |\Covered{p, \variable{examples}}| \geq |\variable{examples}|/k \wedge \text{size}(p) \leq cs^{\alpha}n_t^{\beta}$}{$\variable{result} \gets \variable{result} \cup \{p\}$;}
		}
		\Return $\variable{result}$;
	}
	\Fn{\Search{\variable{examples}, $k$, $n_t$, $s$}}{
		\lIf {$|\variable{examples}| = 0$}{\Return $\{\}$}
		\lIf{(\variable{examples}, $k$) is visited before or $k = 0$} {\Return $\bot$}
		\For{each $p \in \Get{$\variable{examples}, k, n_t,s$}$}{
			\variable{searchResult} $\gets \Search{$\variable{examples} - \Covered{$p, \variable{examples}$}, k - 1, n_t, s$}$; \\
			\lIf{$\variable{searchResult} \neq \bot$}{\Return $\{p\} \cup \variable{searchResult}$}
		}
		\Return $\bot$;
	}
	$s \gets 1;$ \\
	\While {\variable{True}}{
		$n_l \gets cs^{\alpha/(1 - \beta)};\ \ k_l \gets cs\ln |T|;$ \\	
		\For {each $(k, n_t) \in [1, k_l] \times [1, n_l]$}{
			\If {$(k, n_t)$ has not been visited before} {
				$P \gets \Search{$T, k, n_t$}$; \\
				\lIf{$P \neq \bot$}{\Return $P$}
			}
		}
		$s \gets s + 1$;
	}
\end{algorithm}

The algorithm of $\mathcal T_{\text{poly}}$ is comprised of three parts. The first part implements the random sampling discussed previously, as function \Get{} (abbreviated as \G{}). \G{} takes four inputs: $\variable{examples}$ is a set of input-output examples, $k$ is an upper bound on the number of terms, $n_t$ is the number of examples provided to $\mathcal S_t$ and $s$ is an upper bound on the size of terms. Guided by Lemma \ref{lemma:sample}, \G{} returns a set of programs that covers at least $k^{-1}$ portion of examples. The body of \G{} is a repeated sampling process (Line 3). In each turn, $n_t$ examples are sampled (Line 4), and solver $\mathcal S_t$ is invoked to synthesize a program from these sampled examples (Line 5). \G{} collects all valid results (Line 7) and returns them to \Search{} (Line 10).

Note that \G{} only considers those programs of which the sizes are at most $cs^{\alpha}n_t^{\beta}$ (Line 6). This bound comes from the definition of Occam solvers (Definition \ref{definition:occam}): When all examples provided to $\mathcal S_t$ corresponds to the same target term and the size of this term is at most $s$, with a high probability, the term found by $\mathcal S_t$ will be no larger than $cs^{\alpha}n_t^{\beta}$. However, in other cases, the selected examples may happen to correspond to some other unwanted terms: At this time, the term found by $\mathcal S_{t}$ may be much larger than the target term. Therefore, $\mathcal T_{\text{poly}}$ sets a limitation on the size and rejects those programs that are too large. This bound can be safely replaced by any function that is polynomial to $s$ and sub-linear to $n_t$ without affecting $\mathcal T_{\text{poly}}$ to be an Occam term solver.

The second part implements the backtracking as function \Search{} (Lines 11-18). Given a set of examples $\variable{examples}$ and size limit $k$, \Search{} searches for a set of at most $k$ terms that covers all examples. \Search{} invokes function \G{} to obtain a set of possible terms (Line 14), and then recursively tries each of them until a valid term set is found (Lines 15-16).

The third part of $\mathcal T_{\text{poly}}$ selects proper values for $k,n_t$ and $s$ iteratively (Lines 19-29). In each turn, $\mathcal T_{\text{poly}}$ considers the case where the number of target terms and their sizes are all $O(s)$ and select proper values for $n_t$ and $k$ in the following ways:
\begin{itemize}
	\item By Theorem \ref{theorem:occam-error}, when the size of the target term is $O(s)$, $\mathcal S_t$ requires $O(s^{\alpha/(1- \beta)})$ examples to guarantee the accuracy. Therefore, $\mathcal T_{\text{poly}}$ sets the upper bound of $n_t$ to $cs^{\alpha/(1-\beta)}$.
	\item Also by Theorem \ref{theorem:occam-error}, when $n_t$ is set to $cs^{\alpha/(1-\beta)}$, the term syntehsized by $\mathcal S_t$ may still differ with the target term on a constant portion of inputs. As a result, $\mathcal T_{\text{poly}}$ may use $O(\ln n)$ times more terms to cover all examples in $T$. Therefore, $\mathcal T_{\text{poly}}$ sets the upper bound of $k$ to $cs \ln n$.
\end{itemize}
As the time cost of $\G{}$ and $\Search{}$ grows rapidly when $k$ and $n_t$ increases, $\mathcal T_{\text{poly}}$ tries all values of $k$ and $n_t$ from small to large (Lines 22-27), instead of directly using the largest possible $k$ and $n_t$. The iteration ends immediately when a valid term set is found (Line 25).
\subsection{Properties of $\mathcal T_{\text{poly}}$} \label{section:term-property}
In this section, we discuss the properties of $\mathcal T_{\text{poly}}$. As a meta solver, $\mathcal T_{\text{poly}}$ guarantees to be an Occam term solver when $\mathcal S_t$ is an Occam solver on the term space.

\begin{theorem}  \label{theorem:term-occam} $\mathcal S_t$ is an $(\alpha, \beta)$-Occam solver on $T(\mathcal F_C) \implies \mathcal T_{\text{poly}}$ is an $(\alpha' + 1, \beta')$-Occam term solver on $\mathcal F_C$ for any $\alpha' > \alpha, \beta < \beta' < 1$, where $T(\mathcal F_C)$ is defined as $\{(\mathbb P_t, \mathbb I')\ |\ ((\mathbb P_t, \mathbb P_c), \mathbb I) \in \mathcal F_C, \mathbb I' \subseteq \mathbb I\}$.
\end{theorem}

Then, we discuss the time cost of $\mathcal T_{\text{poly}}$. With a high probability, $\mathcal T_{\text{poly}}$ invokes $\Search{}$ only polynomial times, but the time cost of $\Search{}$ may not be polynomial. By Algorithm \ref{alg:term-main}, an invocation of \Search{} of depth $i$ on the recursion tree samples $n_t(k-i)^{n_t}$ times. In the worst case, $\mathcal S_t$ successfully synthesizes programs for all these samples, and the results are all different. At this time, \Search{} will recurse into $n_t(k-i)^{n_t}$ different branches. Therefore, for each $n_t, k$, domain solver $\mathcal S_t$ will be invoked $n_t^k(k!)^{n_t}$ times in the worst case.

However, in practice, $\mathcal T_{\text{poly}}$ is usually much faster than the worst-case because:
\begin{itemize}
	\item The domain solver $\mathcal S_t$ usually fails when the random examples correspond to different target terms, as the expressive ability of term domain $\mathbb P_t$ is usually limited.
	\item For those incorrect terms that happen to be synthesized, they seldom satisfy the requirement on the size and the number of covered examples (Line 6 in Algorithm \ref{alg:term-main}).
\end{itemize}	
In the best case where \G{} never returns a term that is not used by the target program, $\mathcal S_t$ will be invoked at most $n_t2^kk^{n_t}$ times: Such a bound is much smaller than the worst case.
\section{Unifier} \label{section:unifier}
\subsection{Overview}\label{section:unifier-over}
$\mathcal U_{\text{poly}}$ unifies terms into a \textit{decision list}, a structure proposed by \citet{DBLP:journals/ml/Rivest87} for compactly representing decision procedures. $\mathcal U_{\text{poly}}$ unifies a term set $P = \{p_1, \dots, p_m\}$ into the following form:
$$
\texttt{if}\ (c_1)\ \texttt{then}\ p_1\ \texttt{else}\ \texttt{if}\ (c_2)\ \texttt{then}\ p_2\ \texttt{else} \dots \texttt{if}\  (c_{m-1})\ \texttt{then}\ p_{m-1}\ \texttt{else}\ p_m
$$
where $c_1, \dots, c_{m-1}$ belong to $\text{DNF}(\mathbb P_c)$, the disjuctive normal form comprised of conditions in $\mathbb P_c$. $\text{DNF}(\mathbb P_c)$ is defined together with another two sets $\text{L}(\mathbb P_c)$ and $\text{CL}(\mathbb P_c)$, where the set of literals $\text{L}(\mathbb P_c)$ includes conditions in $\mathbb P_c$ and their negations, the set of clauses $\text{CL}(\mathbb P_c)$ includes the conjunctions of subsets of $\text{L}(\mathbb P_c)$, and the set of DNF formulas $\text{DNF}(\mathbb P_c)$ includes disjunctions of subsets of $\text{CL}(\mathbb P_c)$.

We use notation $(\mathbb P_t, \mathbb P_c)_{\text{DL}}$ to denote the set of decision lists where \texttt{if}-terms and \texttt{if}-conditions are from $\mathbb P_t$ and $\text{DNF}(\mathbb P_c)$ respectively. In the following lemma, we show that $(\mathbb P_t, \mathbb P_c)_{\text{DL}}$ is a suitable normal from for designing an Occam unifier, because for any program in $(\mathbb P_t, \mathbb P_c)$, there is always a semantically equivalent program in $(\mathbb P_t, \mathbb P_c)_{\text{DL}}$ with a close size.

\begin{lemma} \label{lemma:decision-tree-size}For any conditional domain $\mathbb D$ and any program $p \in (\mathbb P_t, \mathbb P_c)$, there exists a program $p' \in (\mathbb P_t, \mathbb P_c)_{\text{DL}}$ such that (1) $p'$ is semantically equivalent to $p$ on $\mathbb I$, and (2) $\variable{size}(p') \leq 2 \variable{size}(p)^2$.
\end{lemma}

\SetKwFunction{Covered}{Covered}
\SetKwFunction{SC}{SynthesizeDNF}
\SetKwFunction{AP}{Append}
\SetKwFunction{EP}{Expand}
\begin{algorithm}[t]
	\small
	\caption{The framework of $\mathcal U_{\text{poly}}$.}
	\label{alg:uni-over}
	\LinesNumbered
	\KwIn{A term set $P = \{p_1, \dots, p_m\}$, a PBE task $T$ and a condition solver $\mathcal C$.}
	\KwOut{A program in $(P, \mathbb P_c)$ satisying all examples.}
	\SetKwProg{Fn}{Function}{:}{}
	$\variable{conditionList} \gets \{\}$; \\
	\For {$i \gets 1$; $i < m$; $i \gets i + 1$}{
		$\variable{examples} \gets \{(I, \texttt{false})\ |\ (I, O) \in T - \Covered{$p_i, T$}\}$; \\
		$\variable{examples} \gets \left \{(I, \texttt{true})\ |\ (I, O) \in  \Covered{$p_i, T$} - \bigcup_{j ={i+1}}^{m} \Covered{$p_j, T$} \right \}$; \\
		$c_i \gets \mathcal C(\variable{examples})$; \\
		$\variable{conditionList}.\AP{$c_i$};\ \ T \gets \{(I, O) \in T\ |\ \neg \sem{c_i}(I)\}$; \\
	}
	$\variable{result} \gets p_m$; \\
	\For {$i \gets m - 1$; $i > 0$; $i \gets i - 1$}{
		$\variable{result} \gets \left(\texttt{if}\ \variable{conditionList}_i\ \texttt{then}\ p_i\ \texttt{else}\ \variable{result}\right)$;
	}
	
	\Return \variable{result};
\end{algorithm}

$\mathcal U_{\text{poly}}$ decomposes the unification task into $m-1$ PBE tasks for $c_1, \dots, c_{m-1}$ respectively. Algorithm \ref{alg:uni-over} shows the framework of $\mathcal U_{\text{poly}}$, which synthesizes conditions in order (Lines 2-7). For each term $p_i$ and each remaining example $e=(I, O)$, there are three possible cases:
\begin{itemize} 
	\item If $p_i$ is not consistent with $e$, the value of \texttt{if}-condition $c_i$ must be \texttt{False} on input $I$ (Line 3).
	\item If $p_i$ is the only program in $p_i, \dots, p_n$ that is consistent with $e$, the value of $c_i$ must be \texttt{True} on input $I$ (Line 4). 
	\item Otherwise, the value of $c_i$ does not matter, as $p_i$ is not the last choice. Therefore, $\mathcal U_{\text{poly}}$ ignores this example: If the synthesized $c_i$ is \texttt{false} on input $I$, $e$ will be left to subsequent terms.
\end{itemize}
In this way, $\mathcal U_{\text{poly}}$ obtains a PBE task for $c_i$, and it invokes a DNF solver $\mathcal C$ to solve it (Line 5). Then, $\mathcal U_{\text{poly}}$ excludes all examples covered by $c_i$ (Line 6) and moves to the next term.  At last, $\mathcal U_{\text{poly}}$ unifies all terms and conditions into a complete program (Lines 8-11).

Just like $\mathcal T_{\text{poly}}$, unifier $\mathcal U_{\text{poly}}$ is an Occam unifier when $\mathcal C$ is an Occam solver on $\text{DNF}(\mathbb P_c)$.

\begin{lemma} \label{lemma:unifier-gen} $\mathcal C$ is an $(\alpha, \beta)$-Occam solver on $\text{DNF}(\mathcal F_C) \Rightarrow \mathcal U_{\text{poly}}$ is an $(\alpha', \beta)$-Occam unifier on $\mathcal F_C$ for any $\alpha' > 4\alpha$, where $\text{DNF}(\mathcal F_C)$ is defined as $\{(\text{DNF}(\mathbb P_c), \mathbb I')\ |\ ((\mathbb P_t, \mathbb P_c), \mathbb I) \in \mathcal F_C, \mathbb I' \subseteq \mathbb I\}$. 
	
	$\mathcal C$ is a deterministic $(\alpha, \beta)-$Occam solver on $\text{DNF}(\mathcal F_C) \Rightarrow \mathcal U_{\text{poly}}$ is a $(4\alpha, \beta)-$Occam unifier on $\mathcal F_C$.
\end{lemma}

By Lemma \ref{lemma:unifier-gen}, the only problem remaining is to design an Occam solver for $\text{DNF}(\mathcal F_C)$. We shall gradually build such a condition solver in the following two subsections. 

\subsection{Condition Synthesis for Clauses} \label{section:uni-simple}
For the sake of simplicity, we start by introducing some useful notations:
\begin{itemize}
	\item We regard a DNF formula $d$ as a set of clauses and regard a clause $c$ as a set of literals. We use notation $p_c(c)$ and $p_d(d)$ to represent their corresponding program respectively. 
	\item For an input space $\mathbb I$ and a condition $p$, we use $P(\mathbb I, p)$ and $N(\mathbb I, p)$ to denote the set of inputs where $p$ is evaluated to \texttt{true} and \texttt{false} respectively.
	\item For a PBE task $T$, we use $\mathbb I(T)$, $\mathbb I_P(T)$ and $\mathbb I_N(T)$ to denote the inputs of examples, positive examples and negative examples in $T$ respectively, i.e.:
	$$\mathbb I(T) \coloneqq \{I\ |\ (I, O) \in T\}\ \ \ \mathbb I_P(T) \coloneqq \{I\ |\ (I, \texttt{true}) \in T\}\ \ \ \mathbb I_N(T) \coloneqq \{I\ |\ (I, \texttt{false}) \in T\}$$
\end{itemize}
For a PBE task $T$ on domain $(\text{DNF}(\mathbb P_c), \mathbb I)$, a valid condition $p$ should satisfy the following two conditions: (1) $p$ takes \texttt{true} on all positive examples in $T$, i.e., $\mathbb I_P(T) \subseteq P(\mathbb I(T), p)$; (2) $p$ takes \texttt{false} on all negative examples in $T$, i.e., $\mathbb I_N(T) \subseteq N(\mathbb I(T), p)$. In this section, we build solver $\mathcal C_{\text{CL}}$ for the subproblem where the target condition is assumed to be a single clause. 

\SetKwFunction{SC}{SimplifyClause}
\begin{algorithm}[t]
	\small
	\caption{The pseudo code of clause solver $\mathcal C_{\text{CL}}$.}
	\label{alg:uni-over}
	\LinesNumbered
	\KwIn{A condition space $\mathbb P_c$ and a PBE task $T$.}
	\KwOut{A program in $\text{CL}(\mathbb P_c)$ satisfying all examples or $\bot$.}
	\SetKwProg{Fn}{Function}{:}{}
	\Fn{\SC{$c_u, T$}}{
		$\variable{remNeg} \gets \mathbb I_N(T);\ \ c^* \gets \emptyset;$ \\
		\While {$\variable{remNeg} \neq \emptyset$}{
			$l^* \gets \mathop{\arg\max}_{l \in c_u} \big(|N(\mathbb I(T), l) \cap \variable{remNeg}| / \variable{size}(l)\big)$; \\
			$c^* \gets c^* \cup {l^*}; \ \ \variable{remNeg} \gets \variable{remNeg} - N(\mathbb I(T), l)$;
		}
		\Return $c^*$.
	}
	$c_u \gets \{l \in \text{L}(\mathbb P_c)\ |\ \mathbb I_P(T) \subseteq P(\mathbb I(T), l)\}$; \\
	\lIf {$\mathbb I_N(T) \not \subseteq N(\mathbb I(T), c_u) $}{\Return $\bot$}
	\Return $p_c(\SC{$c_u, T$})$; 
\end{algorithm} 

The pseudo-code of $\mathcal C_{\text{CL}}$ is shown as Algorithm \ref{alg:uni-over}. $\mathcal C_{\text{CL}}$ starts with the first condition of valid clauses: $\mathbb I_P(T) \subseteq P(\mathbb I(T), c)$. By the semantics of operator $\texttt{and}$, for any clause $c$ and any literal $l \in c$, $P(\mathbb I(T), c)$ must be a subset of $P(\mathbb I(T), l)$. Therefore, only those literals that covers all positive examples can be used in the result. $\mathcal C_{\text{CL}}$ collects all these literals as clause $c_u$ (Line 8). Then the subsets of $c_u$ are exactly those clauses satisfying the first condition.

The remaining task is to find a subset $c^*$ of $c_u$ that satisfies the second condition, i.e., $\mathbb I_N(T) \subseteq N(\mathbb I(T), c^*) = \cup_{l \in c^*} N(\mathbb I(T), l)$. Meanwhile, to make $\mathcal C_{\text{CL}}$ an Occam solver, the size of $p_c(c^*)$ should be as small as possible. This problem is an instance of \textit{weighted set covering}: $\mathcal C_{\text{CL}}$ needs to select some sets from $\{N(\mathbb I(T), l)\ |\ l \in c^*\}$ to cover $\mathbb I_N(T)$. This problem is known to be difficult: \citet{DBLP:journals/eccc/Moshkovitz11} proves that for any $\epsilon > 0$, there is no polynomial-time algorithm that always finds a solution at most $(1-\epsilon)\ln n$ times wrose than the optimal, unless $\mathsf {NP} = \mathsf P$, where $n$ is $|\mathbb I_N(T)|$ in our case. 

In our implementation, we use a standard greedy algorithm for weighted set covering, which runs in polynomial time and always finds a solution at most $(\ln n + 1)$ times worse than the optimal (Lines 1-7). $\mathcal C_{\text{CL}}$ maintains set $\variable{remNeg}$, representing the set of negative examples that have not been covered yet (Line 3). In each turn, $\mathcal C_{\text{CL}}$ selects the literal $l^*$ which covers the most uncovered negative examples in each unit of size (Line 5) and includes $l^*$ into the result (Line 6).

The size of the clause found by $\mathcal C_{\text{CL}}$ is bounded, as shown in Lemma \ref{lemma:cl-size}. Therefore, $\mathcal C_{\text{CL}}$ is an Occam solver, as shown in Corollary \ref{corollary:cl-size}. The time complexity of $\mathcal C_{\text{CL}}$ is polynomial to $|\mathbb P_c|$ and $|T|$, and thus $\mathcal C_{\text{CL}}$ is efficient when $\mathbb P_c$ is not large.
 
\begin{lemma} \label{lemma:cl-size} Given condition space $\mathbb P_c$ and PBE task $T$, let $c^*$ be the smallest valid clause and $c$ be the clause found by $\mathcal C_{\text{CL}}$. Then $\text{size}(p_c(c)) < 2\text{size}(p_c(c^*)) (\ln |T| + 1)$.
\end{lemma}

\begin{corollary}\label{corollary:cl-size}
	 For any $0 < \beta < 1$, $\mathcal C_{\text{CL}}$ is an $(1, \beta)$-Occam solver on all possible clause domains.
\end{corollary}
\subsection{Condition Synthesis for Disjunctive Normal Forms} \label{section:uni-full}
In this section, we implement an Occam solver $\mathcal C$ for disjunctive normal forms. By the semantics of operator \texttt{or}, we could obtain a lemma that is similar with Lemma \ref{lemma:sample} in form. 

\begin{lemma}\label{lemma:or-sample} Let $T$ be a PBE task and $d$ be a DNF formula satisfying all examples in $T$, then:
	\begin{itemize}
		\item All clauses in $d$ must be \texttt{false} on all negative examples in $T$, i.e., $\forall c \in d, \mathbb I_N(T) \subseteq N(\mathbb I(T), c)$.
		\item There exists a clause in $d$ that is \texttt{true} on at least $|d|^{-1}$ portion of positive examples in $T$, i.e., $\exists c \in d, |P(\mathbb I(T), c)| \geq |d|^{-1}|\mathbb I_P(T)|$.
	\end{itemize}
\end{lemma}

\SetKwFunction{SR}{Search}
\SetKwFunction{DAC}{DivideAndConquer}
\SetKwFunction{SCS}{SimplifyClauseSet}
\SetKwFunction{GC}{GetPossibleClause}
\SetKwFunction{VC}{IsValidClause}
\SetKwFunction{In}{Insert}
\SetKwFunction{De}{Delete}
\SetKwFunction{G}{Get}
\SetKwFunction{Search}{Search}
\SetKwFunction{GPWS}{GetConditionsWithSizeBound}
By this lemma, $\mathcal C$ can be implemented similarly as $\mathcal T_{\text{poly}}$ by regarding $\mathcal C_{\text{CL}}$ as the domain solver, as shown in Algorithm \ref{alg:uni-dnf-main}. Comparing with the counterpart in $\mathcal T_{\text{poly}}$, there are two main differences:
\begin{itemize}
	\item $\GC{}$ finds a set of clauses that are evaluated to \texttt{false} on all negative examples, and are evaluated to \texttt{true} on at least $k^{-1}$ portion of positive examples (Line 5). Correspondingly, only covered positive examples are excluded in each recursion (Line 6).
	\item For the efficiency of $\mathcal C_{\text{CL}}$, $\mathcal C$ iteratively selects a parameter $s'$ (Line 13). In each iteration, only those literals with size at most $s'$ are available (Line 15).
\end{itemize}



\begin{algorithm}[t]
	\small
	\caption{DNF solver $\mathcal C$ for disjunctive normal forms}
	\label{alg:uni-dnf-main}
	\LinesNumbered
	\KwIn{A condition space $\mathbb P_c$, a PBE task $T$ and a constant $c_0$.}
	\KwOut{A DNF formula in $\text{DNF}(\mathbb P_c)$ satisfying all examples.}
	\SetKwProg{Fn}{Function}{:}{}
	\Fn{\Search{$\variable{literals}, T, k, s$}}{
		\lIf {$|\mathbb I_P(T)| = 0$}{\Return $\{\}$}
		\lIf{$(\variable{literals}, T, k)$ is visited before or $k = 0$} {\Return $\top$}
		\For{each $c \in \GC{$\variable{literals}, T, k$}$}{
			\lIf {$\text{size}(p_c(c)) > c_0s\ln |T|$}{\textbf{continue}}
			$\variable{searchResult} \gets \Search{$\variable{literals}, T - \left\{(I, \texttt{true}) \in T\ |\ \sem{p_c(c)}(I) = \texttt{true} \right\}, k - 1, s$}$; \\
			\lIf{$\variable{searchResult} \neq \bot$}{\Return $\{c\} \cup \variable{searchResult}$}
		}
		\Return $\top$;
	}
	$s \gets 1$; \\
	\While {True}{
		$k_l \gets c_0s$; \\
		\For {each $(k, s') \in [1, k_l] \times [1, s]$}{
			\If {$(k, s')$ has not been visited before}{
				$P_c \gets \GPWS{$\mathbb P_c, s'$}$; \ \ \ $d \gets \Search{$\text{L}(P_c), T, k, s$}$; \\
				\lIf{$d \neq \bot$}{\Return $p_d(d)$}
			}
		}
		$s \gets s + 1$;
	}
\end{algorithm} 


Our implementation of $\GC{}$ (abbreviated as \G{}) optimizes the sampling algorithm in $\mathcal T_{\text{poly}}$ by opening the box of clause solver $\mathcal C_{\text{CL}}$. By Algorithm \ref{alg:uni-over}, $\mathcal C_{\text{CL}}$ synthesizes clauses in two steps. It firstly finds a set $c_u$ of all usable literals and then simplifies it greedily. To synthesize a usable clause, set $c_u$ should satisfy all negative examples and at least $k^{-1}$ portion of positive examples. We find the number of different $c_u$ satisfying this condition is usually small in practice. Therefore, $\G{}$ tries to find all possible $c_u$, and simplifies them using function $\SC{}$.

Given an input space $\mathbb I$ and a set of literals $L$, define relation $\sim_{\mathbb I}$ on clause space $\text{CL}(L)$ as $c_1 \sim_{\mathbb I} c_2 \iff \forall I \in \mathbb I, \sem{p_c(c_1)}(I) = \sem{p_c(c_2)}(I)$, i.e., $\sim_{\mathbb I}$ represents the relation of semantically equivalence on input space $\mathbb I$. Under relation $\sim_{\mathbb I}$, $\text{CL}(L)$ is divided into equivalent classes. We denote the class corresponding to clause $c$ as $[c]_{\mathbb I}$. It is easy to show that each class $[c]_{\mathbb I}$ contains a globally largest clause that is the union of all its elements, i.e., $\exists c' \in [c]_{\mathbb I},  c' = \left(\cup_{x}\ x\text{ for } x \in [c]_{\mathbb I}\right)$. Then, we introduce the concept of \textit{representative clauses} to denote the set of all possible $c_u$ generated by $\mathcal C_{\text{CL}}$.

\begin{definition}[Representative Clauses] 
	Given an input space $\mathbb I$, a size limit $k$, and a set of literals $L$, representative set $R(\mathbb I, k, L) \subseteq \text{CL}(L)$ includes all clause $c$ satisfying: (1) $c$ is the largest clause in $[c]_{\mathbb I}$, (2) $c$ takes \texttt{true} on at least $k^{-1}$ portion of the inputs, i.e., $|P(\mathbb I, c)| \geq k^{-1} |\mathbb I|$.
\end{definition}

\begin{algorithm}[t]
	\small
	\caption{Implementation of \texttt{GetPossibleClause()} in $\mathcal C$.}
	\label{alg:get-clause}
	\LinesNumbered
	\KwIn{A set of literals $\variable{literals}$, a PBE task $T$, and an upper bound $k$.}
	\KwOut{A set of possible clauses according to Lemma \ref{lemma:or-sample}.}
	$\variable{result} \gets \{\emptyset\}$; \\
	\For {each $l \in \variable{literals}$}{
		\For {each $c \in \variable{result}$}{
			\lIf{$\left|P(\mathbb I_P(T), c \cup \{l\})\right| \geq k^{-1}\left|\mathbb I_P(T)\right|$}{$\variable{result}.\In{$c \cup \{l\}$}$}
		}
		\For {each $c \in \variable{result}$}{
			\lIf{$\exists c' \in \variable{result}, \left(P(\mathbb I_P(T), c) = P(\mathbb I_P(T), c') \wedge c \subset c'\right)$}{$\variable{result}.\De{$c$}$}
		}
	}
	\Return $\left\{\SC{c}\ |\  c \in \variable{result} \wedge \mathbb I_N(T) \subseteq N(\mathbb I (T),c)\right\}$
\end{algorithm} 

According to this definition, $R(\mathbb I_P(T), k, L)$ is the set of possible $c_u$ when the PBE task is $T$, the size limit is $k$ and the set of available literals is $L$.  Our implementation of \G{} is shown as Algorithm \ref{alg:get-clause}. \G{} maintains set \variable{result} which is equal to $R(\mathbb I_P(T), k, L')$ for some hidden literal set $L'$. Initially, $L'$ is empty and thus $\variable{result}$ includes only an empty clause, i.e., \texttt{true} (Line 1). \G{} considers all literals in order (Lines 2-9). In each turn, a new literal is inserted to $L'$ and thus \variable{result} is updated correspondingly (Lines 3-8). At last, $\variable{result}$ is equal to $R(\mathbb I_P(T), k, \variable{literals})$, and thus \G{} simplies and returns all valid clauses in $\variable{result}$ (Line 10).

We prove that $\mathcal C$ is an Occam solver, as shown in Lemma \ref{lemma:condition-occam}. Combining with Lemma  \ref{lemma:unifier-gen}, $\mathcal U_{\text{poly}}$ is also an Occam unifier, as shown in Theorem \ref{theorem:unifier-occam}.

\begin{lemma} \label{lemma:condition-occam} For any $0 < \beta < 1$, $\mathcal C$ is a $(2, \beta)$-Occam solver on $\text{DNF}(\mathcal F_C^A)$. 
\end{lemma}

\begin{theorem} \label{theorem:unifier-occam} For any $0 < \beta < 1$, $\mathcal U_{\text{poly}}$ is a $(8, \beta)$-Occam unifier on $\mathcal F_C^A$.
\end{theorem}

Finally, we prove that \mainname is an Occam solver by combining Theorem \ref{theorem:unifier-occam}, Theorem \ref{theorem:term-occam} and Theorem \ref{theorem:gen-stun}, as shown in Theorem \ref{theorem:final}.

\begin{theorem} \label{theorem:final} $\mathcal S_t$ is an $(\alpha, \beta)$-Occam solver on $T(\mathcal F_C)$ with $\beta < \frac{1}{8} \implies$\mainname is an $(8(\alpha' + 1), 8\beta')$- Occam solver on $\mathcal F_C$ for any $\alpha' > \alpha, \beta < \beta' < \frac{1}{8}$.
\end{theorem}

Note that the bound in Theorem \ref{theorem:final} is loose in practice, as it considers all corner cases for the sake of preciseness. For example, while analyzing solver $\mathcal C$, we consider the case where both the number of clauses and the size of the clauses are linear to the size of the target condition $d^*$: At this time, there must be $O(1)$ clauses of which the size is $\Omega(\text{size}(d^*))$, and $\Omega(\text{size}(d^*))$ clauses of which the size is $O(1)$. Such an \texttt{if}-condition is seldom used in practice.

\section{Implementation} \label{section:implementation}
We instantiate our approach \mainname on the domains of \textit{conditional linear integer arithmetic (CLIA)}. Our implementation is in C++, and is available at \url{https://github.com/jiry17/PolyGen}.

Our implementation supports the family of CLIA domains $\mathcal F_I$ defined in the evaluation of the study on \textit{Eusolver}~\cite{DBLP:conf/tacas/AlurRU17}. where  different domains only differ in the number of inputs. Concretely, given the number of inputs $n$, a domain $\mathbb D_I = ((\mathbb P_t, \mathbb P_c), \mathbb I) \in \mathcal F_I$ is defined as follows.
\begin{itemize}
	\item Term space $\mathbb P_t$ contains all linear integer expressions of input variables $x_1, \dots, x_n$.
	\item Condition space $\mathbb P_c$ contains all arithmetic comparisons between linear expressions and their boolean expressions, i.e., $\mathbb P_c$ is the smallest set satisfying the following eqaution.
	\begin{align*}
		\mathbb P_c = \left \{e_1\ \texttt{o}\ e_2\ |\  e_1, e_2 \in \mathbb P_t, \texttt{o} \in \{<, \leq, =\}\right\} &\cup  \left\{c_1 \ \texttt{o}\ c_2\ |\  c_1, c_2 \in \mathbb P_c, \texttt{o} \in \{\texttt{and}, \texttt{or}\}\right\} 
	\cup \left\{\texttt{not}\ c\ |\  c\in \mathbb P_c\}\right\} 
	\end{align*}
By Definition \ref{def:if-closed}, $\mathbb D_I$ is an if-closed conditional domain.
	\item Input space $\mathbb I$ contains integer assignments to input variables. For simplicity, we assume $\mathbb I$ contains all assignments in a bounded range:
	$
		\mathbb I \coloneqq \left\{(w_1, \dots, w_n)\ |\ w_i \in [\text{int}\_\text{min}, \text{int}\_\text{max}]\right\}
	$.
\end{itemize}

To implement $\mathcal T_{\text{poly}}$ and $\mathcal U_{\text{poly}}$, we set parameters $c$ and $c_0$ as $2$. Besides, by Theorem \ref{theorem:final}, an Occam solver $\mathcal S_{t}$ on $T(\mathcal F_{I})$ is required to instantiate \mainname on $\mathcal F_I$. We implement $\mathcal S_t$ as an optimized solver that always synthesizes the smallest valid program. Concretely, given PBE task $T$, $\mathcal S_t$ synthesizes $c_0 + c_1x_1 + \dots + c_nx_n$ by solving the following optimization problem with respect to $c_0, \dots, c_n$:
\begin{align*}
	\textbf{Minimize}\ \text{size}(c_0 + c_1x_1 + \dots + c_n x_n) \ \ \ \textbf{Subject to}\ \forall ((w_1, \dots, w_n), O) \in T, \textstyle{\sum_{i=1}^n} w_ic_i + c_0 = O
\end{align*}
This problem is an instance of \textit{integer linear programming (ILP)}, and  $\mathcal S_t$ solves it by invoking \textit{Gurobi}~\cite{gurobi}, a state-of-the-art solver for ILP. Clearly, $\mathcal S_t$ is an $(1, 0)$-Occam solver, and thus by Theorem \ref{theorem:final}, \mainname is an Occam solver on $\mathcal F_I$. 

\section{Evaluation}\label{section:evaluation}
To evaluate \mainname, we report several experiments to answer the following research questions:

\begin{itemize}
	\item \textbf{RQ1:} How does \mainname compare against existing PBE solvers?
	\item \textbf{RQ2:} How do term solver $\mathcal T_{\text{poly}}$ and unifier $\mathcal U_{\text{poly}}$ affect the performance of \mainname ?
\end{itemize}

\subsection{Experimental Setup}

\noindent \textbf{Baseline Solvers}. We compare \mainname with three PBE solvers, \textit{Esolver}~\cite{DBLP:conf/fmcad/AlurBJMRSSSTU13}, \textit{Eusolver}~\cite{DBLP:conf/tacas/AlurRU17}, and \textit{Euphony}~\cite{DBLP:conf/pldi/LeeHAN18}, which represent the state-of-the-art of three different methods on improving generalizability:
\begin{enumerate}[leftmargin=*]
	\item The first method is guided by the principle of Occam's Razor, which guarantees to synthesize the smallest valid program. On CLIA, \textit{ESolver} is the best-known solver following this method, which enumerates programs in the increasing order of size, and prunes off useless programs via a strategy namely \textit{observational equivalence}.
	\item The second method combines the first method with efficient synthesis techniques heuristically, and thus makes a trade-off between generalizability and efficiency. Among them, \textit{Eusolver} combines the principle of Occam's Razor with the STUN framework by requiring the term solver to enumerate terms in the increasing order of size.
	\item The third method uses a learned model to guide the synthesis. In this category, \textit{Euphony} is the state-of-the-art among solvers available on CLIA. \textit{Euphony} is based on \textit{Eusolver} and uses a model based on structural probability to guide the search of \textit{Eusolver}.
\end{enumerate}

\noindent \textbf{Oracle Models}. Our evaluation follows the framework of OGIS~\cite{DBLP:journals/acta/JhaS17}. We consider two different models of oracles, which cover major usages of PBE solvers in practice.
\begin{enumerate}[leftmargin=*]
	\item In model $\mathbb O_{V}$, the oracle can verify whether a program is correct, and can provide a counter-example if the program is incorrect. To synthesize from these oracles, the framework of \textit{counter-example guided inductive synthesis (CEGIS)}~\cite{DBLP:conf/asplos/Solar-LezamaTBSS06} is usually used. 
	
	Given an oracle $\mathcal O$ in $\mathbb O_{V}$ and a PBE solver $\mathcal S$, we run \textit{CEGIS} with solver $\mathcal S$ to synthesize a program from $\mathcal O$. We measure the generalizability of $\mathcal S$ on $\mathcal O$ as the number of examples finally used by $\mathcal S$ to synthesize a correct program, which is equal to the number of \textit{CEGIS} turns, and we measure the efficiency of $\mathcal S$ as the total time cost of the \textit{CEGIS} framework.
	
	\item In model $\mathbb O_{R}$, the oracle cannot verify the correctness of a program but can provide a set of input-output examples. At this time, a program is usually synthesized by (1) invoking the oracle to generate as many examples as possible under some limits on resource, and then (2) invoking a PBE solver to synthesize a program from these examples.
	
	To evaluate the performance of a PBE solver $\mathcal S$ on an oracle $\mathcal O$ in $\mathbb O_{R}$, we assume that there is a corresponding oracle $\mathcal O'$ in $\mathbb O_{R}$ that could verify whether the synthesized program is completely correct for $\mathcal O$. We run $\mathcal S$ in a similar way as \textit{CEGIS}: starting from an empty set of examples, in each turn, we run $\mathcal S$ on all existing examples. If the synthesis result is verified to be incorrect by $\mathcal O'$, we request a new example from $\mathcal O$ and then start a new turn. We measure the generalizability of $\mathcal S$ on $\mathcal O$ as the total number of used examples. Because the PBE solver is usually invoked only once in practice, we measure the efficiency of $\mathcal S$ as the time cost of the last invocation.
\end{enumerate}

\noindent \textbf{Benchmark}. Our evaluation is conducted on a dataset $\mathcal D$ of $100$ benchmarks. For each benchmark, two different oracles $\mathcal O_V$ and $\mathcal O_R$ are provided, which correspond to models $\mathbb O_V$ and $\mathbb O_R$ respectively. The programs are synthesized in a domain of CLIA as stated in Section~\ref{section:implementation}.
 $\mathcal D$ consists of two parts, $\mathcal D_S$ and $\mathcal D_D$, each obtained from an existing dataset. 

\noindent \textbf{Dataset $\mathcal D_S$}. The first dataset $\mathcal D_S$ consists of 82 benchmarks collected from the general track\footnote{There is also a CLIA track in SyGuS-Comp. We use the dataset of the general track here because (1) all benchmarks in the CLIA track are included in the general track, (2) the general track includes additional benchmarks that are collected from a varies of domains and can also be solved by programs in the CLIA syntax.} in SyGuS-Comp~\cite{DBLP:journals/corr/abs-1904-07146}, where each benchmark is provided with a logic specification $\Phi$. 

To implement two oracles, we apply the algorithm $\mathcal A$ used in \textit{Eusolver}~\cite{DBLP:conf/tacas/AlurRU17}, which could (1) get the correct output for a given input, (2) get a counter-example for an incorrect result:
\begin{itemize} 
	\item {\it Oracle $\mathcal O_V$.} Given a candidate program $p$, $\mathcal O_V$ firstly verifies the correctness of $p$ via an SMT solver, and invokes $\mathcal A$ to generate a counter-example if $p$ is incorrect. 
	\item {\it Oracle $\mathcal O_R$.} $\mathcal O_R$ randomly selects an input and invokes $\mathcal A$ to complete it into an example.
\end{itemize}
$\mathcal A$ is applicable only for special specifications, namely \textit{point-wise}: specification $\Phi$ is point-wise if it only relates an input point to its output. Therefore, we filter out those benchmarks where the specification is not point-wise, and those benchmarks that cannot be solved by a CLIA program. 

\noindent \textbf{Dataset $\mathcal D_D$}. The second dataset $\mathcal D_D$ consists of 18 tasks for synthesizing a combinator in a divide-and-conquer algorithm collected by \citet{DBLP:conf/pldi/FarzanN17}. The synthesized program can be converted to a divide-and-conquer algorithm using \textit{ParSynt}~\cite{DBLP:conf/pldi/FarzanN17}.


For example, 
the following specifies a task for synthesizing a combintor $c$ in the divde-and-conquer algorithm for the maximum segment sum ({\it mss}) problem.
\begin{equation}
	\forall l_1, l_2: \texttt{List},\ c\big(\textit{mss}(l_1), \textit{mss}(l_2), \textit{mps}(l_1), \textit{mps}(l_2), \textit{mts}(l_1), \textit{mts}(l_2)\big) = \textit{mss}(l_1 \cat l_2) \label{equation:spec1}
\end{equation}
where $l_1 \cat l_2$ represents the list concatenation of lists $l_1, l_2$, \textit{mps} represents the maximum prefix sum of a list and \textit{mts} represents the maximum suffix sum of a list. 
In this case, a valid combinator can be obtained from equation $\textit{mss}(l_1 \cat l_2) = \max(\textit{mss}(l_1), \textit{mss}(l_2), \textit{mts}(l_1) + \textit{mps}(l_2))$. 

We choose this dataset because of the following reasons.
\begin{enumerate}
	\item It is a typical application scenario for the oracle model $\mathbb O_R$. On the one hand, it is difficult to verify the correctness of a program, as the specification involves complex list operation that is difficult to model in SMT-Lib. On the other hand, it is easy to collect input-output examples for the combinator, as all inputs and the output are generated by some executable function, as shown in Equation \ref{equation:spec1}.
	\item \texttt{if-then-else} operators are frequently used in the combinator, as there are usually many possible cases when merging two halves. For example, the combinator for \textit{mss} deals with $3$ cases, corresponds to \texttt{if}-terms $\textit{mss}(l_1), \textit{mss}(l_2)$ and $\textit{mts}(l_1) + \textit{mps}(l_2)$ respectively.
	\item Synthesizing the combinator directly is difficult, as it can be rather complex in practice. \textit{ParSynt} can successfully synthesize the combinator only when a program sketch is provided.
\end{enumerate}	

Though it is difficult to verify the correctness of the synthesized program against the specification involving complex list operations, it is not difficult to verify the equivalence of two CLIA programs. The original dataset provides ground truth program $c^*$ for each task, and thus we can still implement the two oracles. 
\begin{itemize}
	\item Given a candidate program $p$, $\mathcal O_V$ uses an SMT solver to verify whether $p$ and $c^*$ is semantically equivalent on the input space.
	\item $\mathcal O_R$ randomly selects an input and runs $c^*$ to get the corresponding output.
\end{itemize} 

\noindent \textbf{Configurations}.  All of the following experiments are conducted on Intel Core i7-8700 3.2GHz 6-Core Processor with 48GB of RAM. We use Z3~\cite{DBLP:conf/tacas/MouraB08} as the underlying SMT solver for oracles in model $\mathbb O_V$, and generate random inputs for oracles in model $\mathbb O_R$ by setting each input variable to a random integer according to a uniform distribution over $[-50, 50]$.

For each execution, we set the time limit as $120$ seconds, the memory limit as $8$ GB, and the example number limit as $10^4$. Besides, as both \mainname and oracles in model $\mathbb O_R$ have randomness, we repeat all related executions $5$ times and consider the average performance only.

\subsection{Exp1: Comparison of Approaches (RQ1)}

\noindent \textbf{Procedure. } For each oracle model, we compare \mainname with \textit{ESolver}, \textit{Eusolver} and \textit{Euphony} on all benchmarks in $\mathcal D$. Among them, the experiment setting for \textit{Euphony} is slightly different from others, as \textit{Euphony} requires a labeled training set. We run \textit{Euphony} in two steps:
\begin{itemize}
	\item First, for those benchmarks in $\mathcal D$ where the target program is not explicitly provided, we label them using the program synthesized by \mainname. 
	\item Second, we run \textit{Euphony} using 3-fold cross-validation. We divide the dataset $\mathcal D$ into three subsets. On each subset, we run \textit{Euphony} with the model learned from the other two subsets. One delicate point is that $\mathcal D$ contains benchmarks that are almost the same except the number of input variables. We put these benchmarks in the same subset and thus avoid data leaks.
\end{itemize}
\noindent \textbf{Results. } The results are summarized in Table \ref{table:exp1} while more details are drawn as Figure \ref{fig:exp1-and-exp}. To compare the generalizability, in each comparison, for each benchmark solved by both \mainname and the baseline solver, we record the ratio of the number of examples used by the baseline solver to the number of examples used by \mainname. The geometric mean of these ratios is listed in column \textit{\#Example}. Similarly, to compare the efficiency, in each comparison, we record the ratio of the time cost of the baseline solver to the time cost of \mainname for those benchmarks solved by both solvers and list the geometric mean of these ratios in column \textit{Time Cost}. 

\begin{table*}[t]
	\caption{The results of comparing \mainname with baselines.} \label{table:exp1}
	\begin{spacing}{1}
		\small
		\begin{tabular}{|p{1.4cm}<{\centering}|p{1.35cm}<{\centering}|p{1.35cm}<{\centering}|p{1.35cm}<{\centering}|p{1.35cm}<{\centering}|p{1.35cm}<{\centering}|p{1.35cm}<{\centering}|p{1.35cm}<{\centering}|}
			\hline
			Model & \multicolumn{3}{c|}{$\mathbb O_V$} & \multicolumn{4}{c|}{$\mathbb O_R$} \\
			\hline 
			Solver & \#Solved & \#Example & Time Cost &   \#Solved & \#Example & \#Example$\geq$ &  Time Cost \\
			\hline
			\mainname & $97$ & $\times 1.000$ &$\times 1.000$ & $93$& $\times 1.000$& \multirow{2}{*}{}& $\times 1.000$ \\
			\cline{1-6} \cline{8-8}
			\textit{Esolver} & $9$  & $\times 0.969$ & $\times 3.668$& $9$ & $\times 1.065$& & $\times 52.271$  \\
			\hline 
			\textit{Eusolver} & $67$ & $\times 2.418$& $\times 7.017$& $65$ & $\times 1.649 $& $\times 3.320$& $\times 13.035$ \\
			\hline
			\textit{Euphony} & $54$ & $\times 2.394$& $\times 9.196$& $53$ & $\times 1.115$ & $\times 3.302$&$\times 15.067$\\
			\hline
		\end{tabular}
	\end{spacing}
\end{table*}

Comparing with \textit{Esolver}, the generalizability of \mainname is almost the same as \textit{ESolver} on both oracle models. Recall that the theory of Occam learning used by \mainname is only an approximation of the principle used by \textit{Esolver}: Occam learning relaxes the requirement from finding the smallest valid program to finding a valid program with a bounded size. The experimental result shows that such an approximation does not affect the generalizability too much in practice. Meanwhile, benefiting from the relaxed requirement on the size provided by Occam Learning, \mainname performs significantly better on efficiency: \mainname solves more than $10$ times benchmarks comparing with $\textit{Esolver}$, with a significant speed-up on those commonly solved benchmarks.
\begin{figure}[t]
	\centering  
	\vspace{-0.35cm} 
	\subfigtopskip=0pt 
	\subfigbottomskip=0pt
	\subfigure[ Results of Exp1 for $\mathbb O_V$ of example]{
		\label{exp1-cegis-example}
		\includegraphics[width=0.30\linewidth]{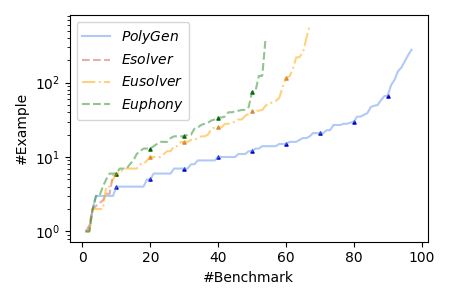}}
	\subfigure[ Results of Exp1 for $\mathbb O_V$ of time]{
		\label{exp1-cegis-time}
		\includegraphics[width=0.30\linewidth]{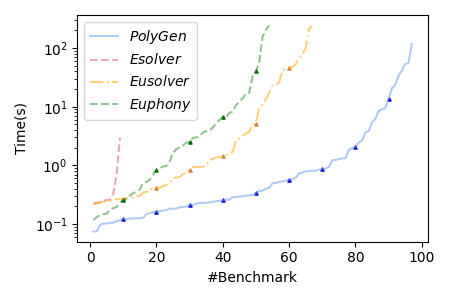}}
	\subfigure[ Results of Exp1 for $\mathbb O_R$ of example]{
		\label{exp1-random-example}
		\includegraphics[width=0.30\linewidth]{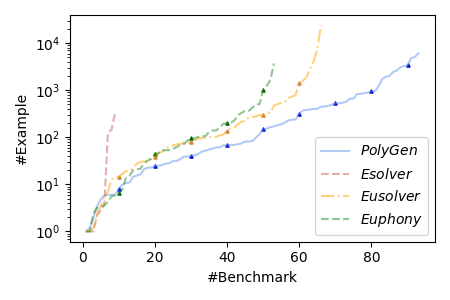}}
	\subfigure[ Results of Exp1 for $\mathbb O_R$ of time]{
		\label{exp1-random-time}
		\includegraphics[width=0.30\linewidth]{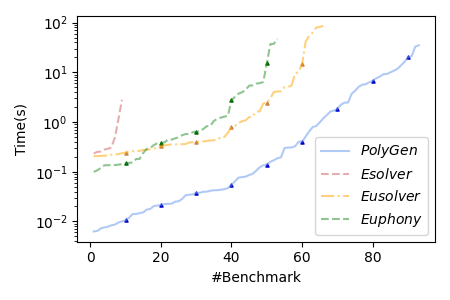}}
	\subfigure[ Results of Exp2 for $\mathbb O_V$ of example]{
		\label{ablation-cegis-example}
		\includegraphics[width=0.30\linewidth]{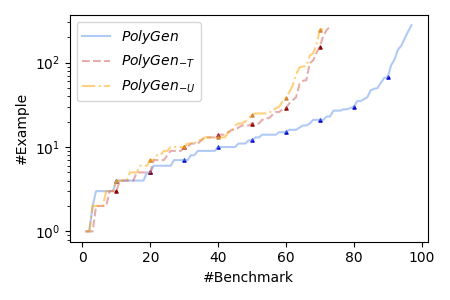}}
	\subfigure[ Results of Exp2 for $\mathbb O_V$ of time]{
		\label{ablation-cegis-time}
		\includegraphics[width=0.30\linewidth]{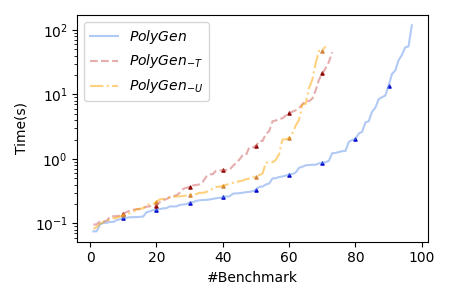}}
	\subfigure[ Results of Exp2 for $\mathbb O_R$ of example]{
		\label{ablation-random-example}
		\includegraphics[width=0.30\linewidth]{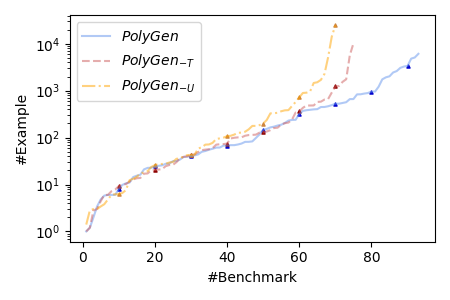}}
	\subfigure[ Results of Exp2 for $\mathbb O_R$ of time]{
		\label{ablation-random-time}
		\includegraphics[width=0.30\linewidth]{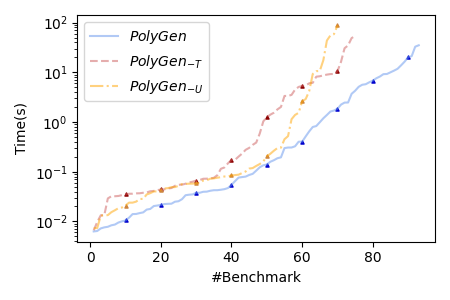}}
	
	\caption{The results of exp1 and exp2. For each approach, we sort its solved benchmarks in the increasing order of the time cost or number of examples needed. Scatters are added for every 10 benchmarks. }
	\label{fig:exp1-and-exp}
\end{figure}

Comparing with \textit{Eusolver} and \textit{Euphony}, \mainname performs significantly better on both generalizability and efficiency. The expreimental result on the generalizability is consistent with our theoretical analysis, as \mainname is an Occam solver but \textit{Eusolver}, \textit{Euphony} are not. Comparing with $\mathbb O_V$, the advantage of \mainname on oracle model $\mathbb O_R$ seems less attractive. The reason is that on $\mathbb O_R$, the cost of distinguishing an incorrect program increases, and thus the effect of synthesis algorithms is weakened. For example, in benchmark \texttt{qm\_neg\_1.sl}, it is hard to distinguish between the target program is $p^*(x) = (\texttt{if}\ (x < 0)\ \texttt{then}\  1\ \texttt{else}\ 0)$ and a wrong program with a slightly different \texttt{if}-condition $p'(x) = (\texttt{if}\ (x \leq 0)\ \texttt{then}\  1\ \texttt{else}\ 0)$ in model $\mathbb O_R$, as the probability for a random input to distinguish them is smaller than $1\%$ when the input is in the range $[-50, 50]$.

Please note that the comparison in Column \textit{\#Example} suffers from a survivorship bias: On a large number of cases that \mainname solves while {\it Eusolver} or \textit{Euphony} does not, \mainname is likely to have better generalizability. To validate this point, we perform an extra experiment on model $\mathbb O_R$. We iteratively rerun \textit{Eusolver} and \textit{Euphony} on those benchmarks where \mainname solves but \textit{Eusolver} or \textit{Euphony} does not. Starting from only $1$ random example, we invoke \textit{Eusolver} or \textit{Euphony} with an enlarged time-limit of 30 minutes. If the solver synthesizes an incorrect program, we record the current number of examples as a lower bound of the generalizability, and then continue to the next iteration by doubling the number of examples. The iteration stops when \textit{Eusolver} or \textit{Euphony} successfully synthesizes a correct program or times out. After adding these lower bounds, the geometric mean of the ratios between the lower bounds and the number of examples used by \mainname is reported in column \textit{\#Example}$\leq$ of Table \ref{table:exp1}. The result justifies the existence of the survivorship bias as the ratios increase from $\times1.115-\times1.649$ to $\times3.320-\times3.302$. Please note that the new experiment still favors to \textit{Eusolver} and \textit{Euphony} because (1) the iteration only provides a coarse lower bound on the number of required examples, (2) the survivorship bias still exists as \textit{Eusolver} and \textit{Euphony} still time out on 27 and 28 out of 93 benchmarks, respectively. 

In terms of efficiency, \mainname solves almost all benchmarks in $\mathcal D$ on both oracle models. We investigate those benchmarks where \mainname times out, and conclude two major reasons:
\begin{itemize}
	\item As the time cost of $\mathcal T_{\text{poly}}$ grows quickly when the number of \texttt{if}-terms increases, \mainname may time out when a large term set is used. For example, \mainname fails in finding a valid term set for \texttt{array\_serach\_15.sl} as this benchmark requires $16$ different \texttt{if}-terms.
	\item As $\mathcal U_{\text{poly}}$ consider conditions in the increasing order of the size, \mainname may time out when a large condition is used. For example, \mainname times out on \texttt{mpg\_example3.sl} which uses \texttt{if}-condition $2x + 2y - z - 7 \leq 0$. This defect can be improved by combining \mainname with techniques on feature synthesis~\cite{DBLP:journals/corr/PadhiM17}.
\end{itemize}

At last, a noteworthy result is that \textit{Euphony} performs even worse than \textit{Eusolver} in our evaluation, implying that the model used in \textit{Euphony} plays a negative role. 
One possible reason is that the target programs in our dataset are diverse and are not easily predictable with a simple probabilistic model considering only the dependency between program elements. 

\subsection{Exp2: Comparison of the Term Solver and the Unifier (RQ2)}

\noindent \textbf{Procedure. } In this experiment, we test how $\mathcal T_{\text{poly}}$ and $\mathcal U_{\text{poly}}$ affect the performance of \mainname. 

Here, we implement two weakened solvers $\mainname_{-T}$ and $\mainname_{-U}$: $\mainname_{-T}$ replaces term solver $\mathcal T_{\text{poly}}$ with the term solver $\mathcal T_{E}$ used in \textit{Eusolver}, and $\mainname_{-U}$ replaces unifier $\mathcal U_{\text{poly}}$ with the unifier $\mathcal U_E$ used in \textit{Eusolver}. For each oracle model, we run these solvers on all benchmarks in $\mathcal D$.

\noindent \textbf{Results. } The results are summarized in Table \ref{table:exp2} while more details are drawn as Figure \ref{fig:exp1-and-exp}.

\begin{table*}[t]
	\caption{The results of comparing \mainname with weakened solvers.} \label{table:exp2}
	\begin{spacing}{1}
		\small
		\begin{tabular}{|p{1.5cm}<{\centering}|p{1.4cm}<{\centering}|p{1.4cm}<{\centering}|p{1.4cm}<{\centering}|p{1.4cm}<{\centering}|p{1.4cm}<{\centering}|p{1.4cm}<{\centering}|}
			\hline
			Model & \multicolumn{3}{c|}{$\mathbb O_V$} & \multicolumn{3}{c|}{$\mathbb O_R$} \\
			\hline 
			Solver & \#Solved & \#Example & Time Cost &   \#Solved & \#Example & Time Cost \\
			\hline
			\mainname & $97$ & $\times 1.000$ &$\times 1.000$ & $93$& $\times 1.000$& $\times 1.000$ \\
			\hline 
			$\mainname_{-T}$ & $73$  & $\times  1.154$ & $\times2.001$& $75$ & $\times0.956$& $\times  2.728$  \\
			\hline 
			$\mainname_{-U}$ & $71$ & $\times 1.644$ & $\times1.95$& $70$ & $\times 1.307 $& $\times  1.943$ \\
			\hline
		\end{tabular}
	\end{spacing}
\end{table*}

As shown in Table \ref{table:exp2}, the unifier $\mathcal U_{\text{poly}}$ improves a lot on both efficiency and generalizability. In contrast, $\mathcal T_{\text{poly}}$ majorly contributes to the efficiency, as the generalizability of \mainname changes little when $\mathcal T_{\text{poly}}$ is replaced. The reason is that the number of examples required by the unifier usually dominates the number of examples required by the term solver, because an example for synthesizing \texttt{if}-conditions, where the output type is Boolean, provides much less information than an example for synthesizing \texttt{if}-terms, where the output is an integer.


	\section{Conclusion}\label{section:conclusion}
	In this paper, we adopt a concept from computational learning theory, Occam Learning, to study the generalizability of the STUN framework. On the theoretical side, we provide a sufficient set of conditions for individual components in STUN to form an Occam solver and prove that \textit{Eusolver}, a state-of-the-art STUN solver, is not an Occam solver. Besides, we design an Occam solver \mainname for the STUN framework. On the practical side, we instantiate \mainname on the domains of CLIA and evaluate it against state-of-the-art PBE solvers on 100 benchmarks and 2 common oracle models. The evaluation shows that (1) \mainname significantly outperforms existing STUN solvers on both efficiency and generalizability, and (2) \mainname keeps almost the same generalizability with those solvers that always synthesize the smallest program, but achieves significantly better efficiency.
	

	\newpage
\clearpage
\appendix
\section{Appendix: Occam Learning} \label{appendix:ocaam learning}
In this section, we compare the difference between the original definition of Occam Learning provided by \citet{DBLP:journals/ipl/BlumerEHW87} and the extended definition used in this paper (Definition \ref{definition:occam}). 

\begin{definition}[Original Definition of Occam Learning] \label{definition:occam-origin} Let $\mathbb C, \mathbb H$ be the concept classes containing target concepts and hypotheses respectively. Then, for constants $\alpha \geq 0$ and $0 \leq \beta < 1$, a learning algorithm $\mathcal L$ is an $(\alpha, \beta)$-Occam algorithm for $\mathbb C$ using $\mathbb H$ iff, given a set $S = \{x_1, \dots, x_m\}$ of $m$ samples labeled according to a concept $c \in \mathcal \mathbb C$, $\mathcal L$ outputs a hypothesis $h \in \mathbb H$ such that:
	\begin{itemize}
		\item $h$ is consistent with $c$ on $S$, i.e., $\forall x \in S, h(x) = c(x)$.
		\item $\text{size}(h) \leq (N \cdot \text{size}(c))^{\alpha}m^{\beta}$.
	\end{itemize} 
where $N$ is the maximum length of any sample $x \in S$, $\text{size}(c)$ and $\text{size}(h)$ are the lengths of the binary representations of concept $c$ and hypothesis $h$ respectively. 
\end{definition}

\begin{definition}[Definition \ref{definition:occam}] \label{definition:occam-ours}For constants $\alpha \geq 1, 0 \leq \beta < 1$, PBE solver $\mathcal S$ is an $(\alpha, \beta)$-Occam solver on a family of domains $\mathcal F$ if there exist constants $c, \gamma > 0$ such that for any program domain $\mathbb D \in \mathcal F$, for any target program $p^* \in \mathbb P$, for any input set $\{I_1, \dots, I_n\} \subseteq \mathbb I$, for any error rate $\epsilon \in \left(0, \frac{1}{2}\right)$:
	$$
	\Pr\left[\text{size}\left(\mathcal S\big(T(p^*, I_1, \dots, I_n)\big)\right) > c\left(\text{size}(p^*)\right)^{\alpha}n^{\beta}\ln^{\gamma}\left(\frac{1}{\epsilon}\right)\right] \leq \epsilon
	$$
	where \text{size}$(p)$ is the length of the binary representation of program $p$, $T(p, I_1, \dots, I_n)$ is defined as the PBE task corresponding to target program $p^*$ and inputs $I_1, \dots, I_n$.
\end{definition}

To adopt the concept Occam Learning to our paper, we firstly replace terms in Definition \ref{definition:occam-origin} with their counterparts in programming by example:
\begin{itemize}
	\item The classes of concepts $\mathbb C$ and hypotheses $\mathbb H$ both correspond to the program space $\mathbb P$.
	\item The learner $\mathcal L$ corresponds to a PBE solver $\mathcal S$.
	\item The target concept $c$, samples $S =\{x_1, \dots, x_m\}$, the set of labeled samples and the learned hypothesis $h$ correspond to the target program $p^*$, the set of inputs $\{I_1, \dots, I_n\}$, the PBE task $T(p^*, I_1, \dots, I_n)$ and the synthesized program $\mathcal S(T(p^*, I_1, \dots, I_n))$ respectively.
\end{itemize}

Second, for simplicity, we assume that the range of an input variable is finite and bounded, and thus ignore the cost of expressing samples, i.e., variable $N$ in Definition \ref{definition:occam-origin}. When the range of an input variable is bounded, a sample can be expressed using $O(k)$ bits, where $k$ is the number of input variables, and thus $N$ can be bounded by $O(\text{size}(p^*))$. Therefore, at this time, removing $N$ from the upper bound would not affect the class of Occam algorithms (solvers), though the concrete values of constants $\alpha$ and $\beta$ may be changed. 

Third, we extend Definition \ref{definition:occam-origin} to support randomness by introduing error rate $\epsilon$. In Definition \ref{definition:occam-ours}, a random PBE solver $\mathcal S$ is allowed to returned a program larger than the upper bound, but the probability should be bounded. In Definition \ref{definition:occam-ours}, factor $\ln(1 /\epsilon)$ is introduced to limit the size of the returned programs concentrate to the original polynomial bound, and thus we could prove a theoretical guarantee on the generalizability of a \textit{random} Occam solver that is similar with the guarantee for deterministic Occam solvers.

\begin{theorem} [Theorem \ref{theorem:occam-error}] Let $\mathcal S$ be an $(\alpha, \beta)$-Occam solver on domain $\mathbb D$. Then there exist constants $c, \gamma >0$ such that for any $0 < \epsilon, \delta < 1$, for any distribution $D$ over $\mathbb I$ and for any target program $p^* \in \mathbb P$:
	$$
	\small
	\forall n > c \left(\frac{1}{\epsilon}\ln\left(\frac{2}{\delta}\right) + \left(\frac{(\text{size}(p^*))^{\alpha}\ln^{\gamma}(2 / \delta)}{\epsilon}\right)^{1/(1-\beta)}\right), \Pr_{I_i \sim D} \left[\text{err}_{D,p^*}\left(\mathcal S\big(T(p^*, I_1, \dots, I_n)\big)\right) \geq \epsilon\right] \leq \delta
	$$
	where $\text{err}_{D,p^*}(p)$ represents the error rate of program $p$ when the input distribution is $D$ and the target program is $p^*$, i.e., $\text{err}_{D,p^*}(p) \coloneqq \Pr_{I\sim D} \left[\sem{p}(I) \neq \sem{p^*}(I)\right]$. 
\end{theorem}
\begin{proof}
	By Definition \ref{definition:occam-ours}, there exists constants $c', \gamma'$ such that:
	\begin{equation}\label{equation:size-delta}
	\Pr\left[\text{size}\left(\mathcal S\big(T(p^*, I_1, \dots, I_n)\big)\right) > c'\left(\text{size}(p^*)\right)^{\alpha}n^{\beta}\ln^{\gamma'}\left(\frac{2}{\delta}\right)\right] \leq \frac{\delta}{2}
	\end{equation}
Let $\mathbb P_{\delta}$ be the set of programs sastisfying that the size is no larger than $c'\left(\text{size}(p^*)^{\alpha}n^{\beta}\ln^{\gamma}(2/\delta)\right)$. By Lemma \ref{lemma:size}, we have an upper bound on $|\mathbb P_{\delta}|$:
\begin{equation*}
\ln |\mathbb P_\delta| \leq c'(\ln2)\left(\text{size}(p^*)^{\alpha}n^{\beta}\ln^{\gamma'}(2/\delta)\right)
\end{equation*}

Let $\mathcal E$ be the event where $\mathcal S$ returns a program outside $\mathbb P_\delta$. Then by Equation \ref{equation:size-delta}, $\Pr[\mathcal E] \leq \frac{\delta}{2}$. Besides, let $\mathbb P_{\epsilon, \delta} \subseteq \mathbb P_{\delta}$ be the set of program $p$ in $\mathbb P_{\delta}$ satisfying $\textit{err}_{D,p^*}(p) \geq \epsilon$. Then:
\begin{align*}
\Pr_{I_i \sim D}\left[\textit{err}_{D,p^*}\left(\mathcal S\big(T(p^*,I_1, \dots, I_n)\big)\right) \geq \epsilon\right] &\leq \Pr_{I_i \sim D}[\mathcal E] + \Pr_{I_i\sim D}\left[\forall p \in \mathbb P_{\epsilon, \delta}, \big(\exists I \in [1,n], \sem{p}(I) \neq \sem{p'}(I)\big)\right] \\
& \leq \frac{\delta}{2} + |\mathbb P_{\delta}| (1 - \epsilon)^{n} 
\end{align*}

Therefore, we obtain the following inequality on $n$:
\begin{align*}
&\Pr_{I_i \sim D}\left[\textit{err}_{D,p^*}\left(\mathcal S\big(T(p^*,I_1, \dots, I_n)\big)\right) \geq \epsilon\right] \leq \delta \\\Longleftarrow&\ |\mathbb P_{\delta}|(1-\epsilon)^n \leq \frac{\delta}{2} \\
 \Longleftarrow &\ n \ln \left(\frac{1}{1 - \epsilon}\right) n \geq \ln\left(\frac{2}{\delta}\right) + c'(\ln 2)\left(\text{size}(p^*)^{\alpha}n^{\beta}\ln^{\gamma'}\left(\frac{2}{\delta}\right)\right) \\
 \Longleftarrow &\ n\geq \frac{1}{\epsilon}\ln\left(\frac{2}{\delta}\right) + \frac{c'(\ln 2)}{\epsilon}\left(\text{size}(p^*)^{\alpha}n^{\beta}\ln^{\gamma'}\left(\frac{2}{\delta}\right)\right) \\
  \Longleftarrow &\ n\geq c \left(\frac{1}{\epsilon}\ln\left(\frac{2}{\delta}\right) + \left(\frac{(\text{size}(p^*))^{\alpha}\ln^{\gamma}(2 / \delta)}{\epsilon}\right)^{1/(1-\beta)}\right)
\end{align*}
where $c$ is a large enough constant and $\gamma = \gamma'$.
\end{proof}

\section{Appendix: Proofs} \label{appendix:proofs}
In this section, we complete the proofs of the theorems in our paper.

\begin{lemma} [Lemma \ref{lemma:size}]For any domain $\mathbb D$, $\forall p \in \mathbb P$, 
	$ \left|\left\{p' \in \mathbb P\  |\ \text{size}(p') \leq \text{size}(p)\right\} \right| \leq 2^{\text{size}(p)}
	$.
\end{lemma} 
\begin{proof}
	Let $n$ be the number of grammar rules used to derive program $p$, $R$ be the set of available grammar rules, and $N$ be the size of $R$. Let $r^*$ be an arbitrary rule in $R$. 
	
	Let $\mathbb P_p$ be the set of all program $p'$ satisfying $\text{size}(p') \leq \text{size}(p)$. Define function $\varphi: \mathbb P_p \mapsto R^n$ as $\varphi(p) \coloneqq r_1\dots r_{n'}r^*\dots r^*$, where $r_1 \dots r_{n'}$ is the leftmost derivation of program p', and $r^*$ repeats for $n-n'$ times. Clearly, $\varphi$ is an injection, i.e., $\forall p_1, p_2 \in \mathbb P_p, p_1 \neq p_2 \implies \varphi(p_1) \neq \varphi(p_2)$.
	
	Therefore, $|\mathbb P_p| \leq |R^n| \leq N^n \leq 2^{\text{size}(p)}$.
\end{proof}

\begin{theorem}[Theorem \ref{theorem:min-rand}] Let $\mathcal F^A$ be the family of all possible domains. Then $\mathcal S_{\text{min}}$ is an $(1, 0)$-Occam solver on $\mathcal F^A$, and $\mathcal S_{\text{rand}}$ is not an Occam solver on $\mathcal F^A$.
\end{theorem}
\begin{proof}
	We start with $\mathcal S_{\text{min}}$. Let $p^*$ be the target program and $p$ be the program synthesized by $\mathcal S_{\text{min}}$. By the definition of $\mathcal S_{\text{min}}$, $\text{size}(p) \leq \text{size}(p^*)$. Therefore, $\mathcal S_{\text{min}}$ is an $(1, 0)$-Occam solver. 
	
	For $\mathcal S_{\text{rand}}$, suppose that all programs in the program space satisfy all given examples, and the target program $p^*$ is the smallest program in the program space. Let $p$ be the program synthesized by $\mathcal S_{\text{rand}}$. Then, by the definition of $\mathcal S_{\text{rand}}$, $p$ follows a uniform distribution on the program space. Therefore, $\text{size}(p)$ can be arbitrarily larger than $\text{size}(p^*)$, and $\mathcal S_{\text{rand}}$ is not an Occam solver.
\end{proof}

\begin{theorem} [Theorem \ref{theorem:gen-stun}] Let $\mathcal F_C$ be a family of if-closed conditional domains, $\mathcal T$ be an $(\alpha_1, \beta_1)$-Occam term solver on $\mathcal F_C$, $\mathcal U$ be an $(\alpha_2, \beta_2)$-Occam unifier where $\beta_1\alpha_2 + \beta_2< 1$. Then the STUN solver comprised of $\mathcal T$ and $\mathcal U$ is an $((\alpha_1 + 1)\alpha_2, \beta_1\alpha_2 +\beta_2)$-Occam solver.
\end{theorem}
\begin{proof} Let $p^*$ be the target program, and $P = \{t_1, \dots, t_n\}$ be the term set synthesized by $\mathcal T$. By the definition of an Occam term solver, there exists constants $c_1$ and $\gamma_1$ such that:
	$$
	\forall \epsilon_1 \in \left(0, \frac{1}{2}\right), \Pr\left[\text{tsize}\left(P\right) > c_1\left(\text{size}(p^*)\right)^{\alpha_1}n^{\beta_1}\ln^{\gamma_1}\left(\frac{1}{\epsilon_1}\right)\right] \leq \epsilon_1
	$$

We construct the following function $\varphi: \mathbb P_t \mapsto (P, \mathbb P_c)$ that converts a term into the program space of the unifier $\mathcal U$:
$$
\varphi(t) \coloneqq \texttt{if }(t = t_1) \texttt{ then } t_1 \texttt{ else }\dots \texttt{ else if }(t = t_{n-1})\texttt{ then }t_{n-1} \texttt{ else }t_n 
$$
By the definition of $\text{size}(p)$, $\text{size}(\varphi(t)) \leq 2\lceil \log_2 N \rceil + n \cdot \text{size}(t) + 2\text{tsize}(P) \leq c_3\text{size}(t)\text{tsize}(P)$, where $N$ is the number of grammar rules, $c_3$ is a large enough constant.

$\varphi$ can be extend to the whole program space $(\mathbb P_t, \mathbb P_c)$ where $\varphi(p)$ is the program replacing all terms $t$ used in $p$ with $\varphi(t)$. Clearly, $\text{size}(\varphi(p))$ is no larger than $c_3\text{size}(p)\text{tsize}(P)$.

Let $p_u$ be the program synthesized by $\mathcal U$. As $\varphi(p^*)$ is a valid program for the unifier $\mathcal U$, by the definition of  an Occam unifier, there exists constants $c_2$ and $\gamma_2$ such that:
$$
\forall \epsilon_2 \in \left(0, \frac{1}{2}\right), 	\Pr\left[\text{size}\left(p_u\right) > c_2\big(\max(\text{size}(\varphi(p^*)), \text{tsize}(P))\big)^{\alpha_2}n^{\beta_2}\ln^{\gamma_2}\left(\frac{1}{\epsilon_2}\right)\right] \leq \epsilon_2
$$

For any $\epsilon \in \left(0, \frac{1}{2}\right)$, by taking $\epsilon_1 = \epsilon_2 = \frac{1}{2}\epsilon$, we obtain the following inequality:
\begin{align*}
	&\Pr\left[\text{tsize}\left(P\right) > c_1\left(\text{size}(p^*)\right)^{\alpha_1}n^{\beta_1}\ln^{\gamma_1}\left(\frac{2}{\epsilon}\right)\right] \leq \frac{\epsilon}{2} \bigwedge \\
&\Pr\left[\text{size}\left(p_u\right) > c_2\big(\max(\text{size}(\varphi(p^*)), \text{tsize}(P))\big)^{\alpha_2}n^{\beta_2}\ln^{\gamma_2}\left(\frac{2}{\epsilon}\right)\right] \leq \frac{\epsilon}{2} \\
\implies &\Pr\left[\text{size}\left(p_u\right) > c_2\left(c_1c_3\left(\text{size}(p^*)\right)^{\alpha_1+1}n^{\beta_1}\ln^{\gamma_1}\left(\frac{2}{\epsilon}\right)\right)^{\alpha_2}n^{\beta_2}\ln^{\gamma_2}\left(\frac{2}{\epsilon}\right)\right] \leq \epsilon \\
\implies & \Pr\left[\text{size}\left(p_u\right) > c\left(\text{size}(p^*)\right)^{(\alpha_1+1)\alpha_2}n^{\beta_1\alpha_2 + \beta_2}\ln^{\gamma}\left(\frac{1}{\epsilon}\right)\right] \leq \epsilon 
\end{align*}
where $c$ is a large enough constant and $\gamma = \gamma_1 \alpha_2 + \gamma_2$. 

Therefore, the STUN solver comprised of $\mathcal T$ and $\mathcal U$ is an $((\alpha_1 + 1)\alpha_2, \beta_1\alpha_2 +\beta_2)$-Occam solver.
\end{proof}

\begin{theorem}[Theorem \ref{theorem:eusolver-term}] $\mathcal T_E$ is not an Occam term solver on $\mathcal F^A_C$, and \textit{Eusolver} is not an Occam solver on $\mathcal F^A_C$, where $\mathcal F^A_C$ is the family of all if-closed conditional domains.
\end{theorem}
\begin{proof} This theorem is directly from Example \ref{example:eusolver-size}.
\end{proof}

\begin{theorem}[Theorem \ref{theorem:eusovler-unifier}] There is no polynomial-time Occam unifier on $\mathcal F^A_C$ unless $\mathsf{NP} = \mathsf{RP}$.
\end{theorem}
\begin{proof} Suppose there is a polynomial-time Occam unifier $\mathcal U$. Given a decision tree learning problem with $k$ different tests and $n$ data $d_1, \dots, d_n$ labeled with $m$ different labels. We construct a conditional program domain $\mathbb D$ where:
	\begin{itemize}
		\item Input space $\mathbb I$ contains $n$ different inputs, corresponding to data $d_1, \dots, d_n$.
		\item Term space $\mathbb P_t$ contains $m$ different constants $1, 2, \dots, m$, corresponding to $m$ different labels. 
		\item Condition space $\mathbb P_c$ contains $k+1$ different conditions $\texttt{false}, c_1, \dots, c_k$, where $c_1, \dots, c_k$ correspond to $k$ different tests. If the test result of the $i$th test on data $d_j$ is \texttt{true}, $\sem{c_i}(d_j)$ will be defined as \texttt{true}. Otherwise, $\sem{c_i}(d_j)$ will be defined as \texttt{false}.
	\end{itemize} 
Because $\forall a, b \in \mathbb P_t, \forall d_i \in \mathbb I, \sem{\texttt{false}}(d_i) \iff \sem{a}(d_i) = \sem{b}(d_i)$, $\mathbb D$ is an if-closed domain.

Let $T = \{(d_i, l_i)\}$ be a PBE task where $l_i$ is the index of the label corresponding to data $d_i$, and let $p$ be the program synthesized by $\mathcal U$ for PBE task $T$ and term set $\{1,2, \dots, m\}$. We remove the usages of condition $\texttt{false}$ in $p$ by replacing program ($\texttt{if (false) then } p_1 \texttt{ else }p_2$) with its \texttt{else}-branch $p_2$. Suppose $p'$ be the resulting program. Clearly, the size of $p'$ is no larger than $p$ and $p'$ can be mapped back into a decision tree.

So far, we construct a polynomial-time learning algorithm for decision trees based on $\mathcal U$. By Theorem \ref{theorem:occam-error}, this construction implies that the class of decision trees is PAC learnable. Combining with the fact that the class of decision trees is not PAC learnable unless $\mathsf{NP} = \mathsf{RP}$~\cite{DBLP:conf/stacs/HancockJLT95}, we know there is no polynomial-time Occam unifier unless $\mathsf{NP} = \mathsf{RP}$.
 \end{proof}

\begin{lemma} [Lemma \ref{lemma:term-split}] For constants $\alpha_1,\alpha_2 \geq 0, 0 \leq \beta_1, \beta_2 < 1$ where $\beta_1 + \beta_2 < 1$, term solver $\mathcal T$ is an $(\alpha_1 + \alpha_2, \beta_1 + \beta_2)$-Occam solver on $\mathcal F_C$ if there exist constants $c, \gamma > 0$ such that for any conditional domain $\mathbb D \in \mathcal F_C$, any target program $p^* \in \mathbb P$, and any input set $\{I_1, \dots, I_n\} \subseteq \mathbb I$:
	\begin{enumerate}
		\item With a high probability, the size of terms returned by $\mathcal T$ is bounded by $\text{size}(p^*)^{\alpha_1} n^{\beta_1}$.
		$$
		\Pr\left[\max \left\{ \text{size(p)}\ \big |\ p \in \mathcal T \big((I_1, \sem{p}(I_1)), \dots, (I_n, \sem{p}(I_n))\big)\right\}  > c\left(\text{size}(p^*)\right)^{\alpha_1}n^{\beta_1}\ln^{\gamma}\left(\frac{1}{\epsilon}\right)\right] \leq \epsilon
		$$
		\item With a high probability, the number of terms returned by $\mathcal T$ is bounded by $\text{size}(p^*)^{\alpha_2} n^{\beta_2}$.
		$$
		\Pr\left[\left|\mathcal T \big((I_1, \sem{p}(I_1)), \dots, (I_n, \sem{p}(I_n))\big)\right|  > c\left(\text{size}(p^*)\right)^{\alpha_2}n^{\beta_2}\ln^{\gamma}\left(\frac{1}{\epsilon}\right)\right] \leq \epsilon
		$$
	\end{enumerate}
\end{lemma}
\begin{proof} Let $P$ be the synthesized term set. For any $\epsilon \in \left(0, \frac{1}{2}\right)$, with a probability of at least $1 - \epsilon$:
	\begin{align*}
	\max \left\{\text{size(p)}\ \big |\ p \in \mathcal P\right\} \leq c\left(\text{size}(p^*)\right)^{\alpha_1}n^{\beta_1}\ln^{\gamma}\left(\frac{2}{\epsilon}\right) \bigwedge 
	\left|P\right|  \leq c\left(\text{size}(p^*)\right)^{\alpha_2}n^{\beta_2}\ln^{\gamma}\left(\frac{2}{\epsilon}\right)
	\end{align*}

Because $\text{tsize}(P) \leq |P| \cdot \max \left\{\text{size}(p)\ |\ p \in P\right\}$, with a probability of at least $1-\epsilon$:
\begin{align*}
\text{tsize}(P) &\leq \left(c\left(\text{size}(p^*)\right)^{\alpha_1}n^{\beta_1}\ln^{\gamma}\left(\frac{2}{\epsilon}\right)\right) \cdot \left(c\left(\text{size}(p^*)\right)^{\alpha_2}n^{\beta_2}\ln^{\gamma}\left(\frac{2}{\epsilon}\right)\right) \\
&\leq c'\left(\text{size}(p^*)\right)^{\alpha_1 + \alpha_2}n^{\beta_1 + \beta_2}\ln^{2\gamma}\left(\frac{1}{\epsilon}\right)
\end{align*}
where $c'$ is a large enough constant. Therefore, $\mathcal T$ is a $(\alpha_1 + \alpha_2, \beta_1 + \beta_2)$-Occam solver.
\end{proof}

\begin{lemma} [Lemma \ref{lemma:sample}]Let $T$ be a PBE task, and let $P$ be a set of terms that covers all examples in $T$, i.e., $\forall (I, O) \in T, \exists p \in P, \left(\sem{p}(I) = O\right)$. There is always a term $p \in P$ such that:
	$$
	\left|\big\{(I, O) \in T\ \big|\ \sem{p}(I) = O\big\} \right| \geq |T|/|P|
	$$
\end{lemma}
\begin{proof} Let $t_1, \dots, t_n$ be the terms in $P$, and $w_1, \dots, w_n$ be the number of examples covered by term $t_i$, i.e., $w_i \coloneqq \left|\left\{(I, O) \in T\ |\ \sem{t_i}(I) = O\right\}\right|$. As $P$ covers all examples in $T$, $\sum_{i=1}^n w_i$ must be at least $|T|$. Therefore, we have $\max w_i \geq |T| / |P|$.
\end{proof}

\begin{theorem}  [Theorem \ref{theorem:term-occam}] $\mathcal S_t$ is an $(\alpha, \beta)$-Occam solver on $T(\mathcal F_C) \implies \mathcal T_{\text{poly}}$ is an $(\alpha' + 1, \beta')$-Occam term solver on $\mathcal F_C$ for any $\alpha' > \alpha, \beta < \beta' < 1$, where $T(\mathcal F_C)$ is defined as $\{(\mathbb P_t, \mathbb I')\ |\ ((\mathbb P_t, \mathbb P_c), \mathbb I) \in \mathcal F_C, \mathbb I' \subseteq \mathbb I\}$.
\end{theorem}
\begin{proof} Let $s^*$ be the value of variable $s$ (Line 19 in Algorithm \ref{alg:term-main}) when $\mathcal T_{\text{poly}}$ terminates. Let $p^*$ be the target program, $P^*$ be the terms used by $\mathcal T_{\text{poly}}$, $P$ be the term set synthesized by $\mathcal T_{\text{poly}}$, and $n$ be the number of examples, i.e., $|T|$. By Algorithm \ref{alg:term-main}, the total size of $P$ is bounded by $s^*$:
	$$
	\text{tsize}(P) \leq |P| \cdot \max_{p \in P} \text{size}(p) \leq  cs^* \log n \cdot c(s^*)^{\alpha}n^{\beta} \leq c'(s^*)^{\alpha+1}n^{\beta''}
	$$
where $\beta''$ is a constant larger than $\beta$ and $c'$ is a large enough constant. 
\SetKwFunction{Search}{Search}
\SetKwFunction{Covered}{Covered}

Let $c_1$ be a large enough constant, $\alpha'$ be any constant larger than $1
$, $\beta'$ be any constant larger than $0$. Suppose $s \geq c_1\text{size}(p^*)^{\alpha'}n^{\beta'}$, $n_t = cs^{\alpha/(1-\beta)}$, $k_t = cs \log n$. We denote an  invocation \Search{$T', k, n_t, s$} valid if the following three random events happens:
\begin{itemize}
	\item $\mathcal E_1$: There is at least one valid sampling round, where a sampling round is valid if $t^*$ covers all sampled examples, and $t^*$ is the target term in $P^*$ that covers the most examples in $T'$.
	\item $\mathcal E_2$: $\text{size}(t_0) \leq cs^{\alpha}n_t^{\beta}$, where $t_0$ is the term synthesized in the first valid sampling round.
	\item $\mathcal E_3$: $\text{size}(t_0)$ covers at least a half of those examples covered by $t^*$ in $T'$, i.e., 
	$$\left(\big|\Covered{$t_0, T'$} \cap \Covered{$t^*, T'$}\big| \geq \frac{1}{2}\big|\Covered{$t^*,T'$}\big|\right) \wedge \left(\big|\Covered{$t^*,T'$}\big| > 0\right)$$
\end{itemize}

Consider a recursion chain $RC$ of function \Search{} with examples $T'_0 = T, \dots, T'_{n_c} = \emptyset$, where $T_i$ recurses into $T_{i+1}$ by including $t_0$ in the result, i.e., $T_{i+1} = T_{i} - \Covered{$p^*, T_{i}$}$.
\begin{itemize}
	\item \textbf{Claim:} When all these invocations are valid, $n_c$ is at most $2|P^*|\ln |T| + 1$.
	\item \textbf{Proof:} By the definition of $t^*$, $|\Covered{$t^*, T_{i}$}| \geq |T_i|/|P^*|$. Then by the definition of $\mathcal E_3$, we have $|T_{i+1}| \leq \left(1 - 1/(2|P^*|)\right)|T_i|$. Therefore, $n_c \leq 2|P^*|\ln |T| + 1$.
\end{itemize}

Because $k_t = cs\log n > 2 |P^*|\ln n$, such a recursion chain is allowed. Therefore, we only need to consider the probability for all invocations to be valid:
\begin{itemize}
	\item $\mathcal E_1$. As $t^*$ is the target term that covers the most examples, the probability for a random example to be in $\Covered{$t^*$, T'}$ is at least $1/|P^*|$. Therefore, the probability for a sampling round to be valid is $|P^*|^{-n_t}$:
	$$
	\Pr[\neg \mathcal E_1] \leq \left(1 - |P^*|^{-n_t}\right)^{n_tk^{n_t}} \leq \exp(-n_t) < \frac{1}{12\text{size}(p^*)\ln n}
	$$ 
	The last two inequalities hold when constant $c_1$ is large enough.
	\item $\mathcal E_2$. As $\mathcal S_t$ is an $(\alpha, \beta)$-Occam solver, there exists constants $c_2, \gamma$ such that:
	$$
	\forall \epsilon \in \left(0, \frac{1}{2}\right), \Pr\left[\text{size}(t_0) > c_2
	(\text{size}(p^*))^{\alpha}n_t^{\beta}\ln^{\gamma}\left(\frac{1}{\epsilon}\right)\right] < \epsilon
	$$
	Therefore, with a probability at least $1 - 1/(12\text{size}(p^*)\ln n))$, we have:
	$$
	\text{size}(t_0) \leq c_2
	(\text{size}(p^*))^{\alpha}n_t^{\beta}\ln^{\gamma}(12\text{size}(p^*)\ln n) < cs^{\alpha}n_t^{\beta} 
	$$
	The last inequality holds when constant $c_1$ is large enough. 
	\item $\mathcal E_3$. Let $\mathbb I'$ be the input space that contains all inputs in $\Covered{$t^*, T'$}$, and $D$ be a uniform distribution on $\mathbb I'$. As $\mathcal S_t$ is an $(\alpha, \beta)$-Occam solver, by Theorem \ref{theorem:occam-error}, $\mathcal E_3$ happens with a probability of at least $1 - 1/(12\text{size}(p^*)\ln n)$ if the following inequality holds:  
$$n_t > c_3 \left(\frac{1}{2}\ln\big(24\text{size}(p^*) \ln n\big) + \left(\frac{(\text{size}(p^*))^{\alpha}\ln^{\gamma}\big(24\text{size}(p^*) \ln n\big)}{\epsilon}\right)^{1/(1-\beta)}\right)$$
where $c_3$ is a fixed constant. Clearly, when $c_1$ is large enough, this inequality always holds. 	
\end{itemize}

Let $\mathcal E_i^{S}$ be the random event that the $i$th invocation in $RC$ is valid, and $\mathcal E^{RC}$ be the random event that all events in $RC$ are valid. Then:
\begin{align*}
\Pr\left[\neg \mathcal E^{RC}\right] &\leq \sum_{i=1}^{n_c} \Pr\left[\mathcal E^{S}_i\right] \leq (2|P^*|\ln n + 1) \big(\Pr[\neg \mathcal E_1] + \Pr[\neg \mathcal E_2] + \Pr[\neg \mathcal E_3]\big) \\
& \leq (2\text{size}(p^*)\ln n + 1) \cdot \frac{1}{4\text{size}(p^*)\ln n} \leq \frac{3}{4}
\end{align*}

Therefore, when $s$ is larger than $c_1\text{size}(p^*)^{\alpha'}n^{\beta'}$, in each iteration, $\mathcal T_{\text{poly}}$ returns with a probability at least $\frac{1}{4}$. Therefore, for any $\epsilon \in (0, 1)$:
\begin{align*}
&\Pr\left[s^* \leq c_1\text{size}(p^*)^{\alpha'}n^{\beta'} + c_4 \ln \left(\frac{1}{\epsilon}\right)\right] < \epsilon \\
\implies& \Pr\left[\text{tsize}(P^*) \leq c_5\text{size}(p^*)^{\alpha'(1 + \alpha)}n^{\beta''+(1+\alpha)\beta'}\ln \left(\frac{1}{\epsilon}\right)\right] < \epsilon
\end{align*}
where $c_4$ and $c_5$ are large enough constants. Note that $\alpha'$ is an arbitrary constant larger than $1$, $\beta'$ and $\beta''$ are arbitrary constants larger than $0$. Therefore, $\mathcal T_{\text{poly}}$ is an $(\alpha^*, \beta^*)$-Occam term solver for any $\alpha^* > 1 + \alpha, \beta^* > 0$.
\end{proof}

\begin{lemma} [ Lemma \ref{lemma:decision-tree-size}]For any conditional domain $\mathbb D$ and program $p \in (\mathbb P_t, \mathbb P_c)$, there is a program $p' \in (\mathbb P_t, \mathbb P_c)_{\text{DL}}$ s.t. (1) $p'$ is semantically equivalent to $p$ on $\mathbb I$, and (2) $\variable{size}(p') \leq 2 \variable{size}(p)^2$.
\end{lemma}
\begin{proof} Let $P = \{t_1, \dots, t_m\}$ be the set of terms used in $p$. Without loss of generality, assume the structure of $p'$ is as the following:
$$
\texttt{if}\ (c_1)\ \texttt{then}\ t_1\ \texttt{else}\ \texttt{if}\ (c_2)\ \texttt{then}\ t_2\ \texttt{else} \dots \texttt{if}\  (c_{m-1})\ \texttt{then}\ t_{m-1}\ \texttt{else}\ t_m
$$

For each \texttt{if}-term in $p$, we define its tree path as a sequence $(c_1, k_1), \dots, (c_n, k_n)$, where $c_i \in \mathbb P_c$ represents the \texttt{if}-conditions corresponding to this term from the top down, and $k_i \in \{0, 1\}$ represents the \texttt{if}-branches taken by this term ($0, 1$ represent the \texttt{then}-branch and the \texttt{else}-branch respectively). Let $\varphi$ be a function mapping a path to a condition, which is defined as the following:
$$
\varphi((c_1, k_1), \dots, (c_n ,k_n)) \coloneqq \left(\texttt{and}_{k_i = 0}\ c_i \right)\texttt{ and }\left(\texttt{and}_{k_i = 1} \ (\texttt{not }c_i) \right)
$$

We construct the condition $c_i$ from the original program $p$. Let $X_i$ be the set of paths corresponding to all usages of term $t_i$. Then we construct the corresponding condition $c_i$ as $\texttt{or}_{x \in X}\ \varphi(x)$.

Clearly, $p'$ is semantically equivalent to $p$ on the input space $\mathbb I$. Besides, $c_1, \dots, c_{m-1}$ are all DNF formulas, and thus $p' \in (\mathbb P_t, \mathbb P_c)_{\text{DL}}$. Let $c_1', \dots, c_{n_c}'$ be all \texttt{if}-conditions used in $p$, $l_i$ and $r_i$ be the numbers of terms used in the \texttt{then}-branch and the \texttt{else}-branch of $c_i$ respectively. Then, $\text{size}(p')$ can be calculated in the following way:
\begin{itemize}
	\item The total size of \texttt{if}-conditions $c_i$ is at most: $$\sum_{i=1}^{n_c} \big(l_i(\text{size}(c'_i) + \lceil \log_2 N \rceil) + r_i(\text{size}(c'_i) + 2\lceil \log_2 N \rceil)\big) - (m - 1)\lceil \log_2 N \rceil$$
	where $N$ is the number of grammar rules.
	\item The total size of \texttt{if}-operators is $(m-1)\lceil \log_2 N \rceil$, the total size of \texttt{if}-terms is $\sum_{i=1}^m \text{size}(t_i)$.
\end{itemize}
Therefore, we have the following inequality:
\begin{align*}
\text{size}(p') \leq& \sum_{i=1}^{n_c} \big(l_i(\text{size}(c'_i) + \lceil \log_2 N \rceil) + r_i(\text{size}(c'_i) + 2\lceil \log_2 N \rceil)\big) \\&- (m - 1)\lceil \log_2 N \rceil + (m-1)\lceil \log_2 N \rceil + \sum_{i=1}^m \text{size}(t_i) \\
\leq& \left(\sum_{i=1}^{n_c} (l_i + r_i)\text{size}(c_i') + \sum_{i=1}^m \text{size}(t_i)\right) + 2\lceil \log_2 N \rceil\sum_{i-1}^{n_c} (l_i + r_i) \\
\leq & \text{size}(p)\left(\sum_{i=1}^{n_c}\text{size}(c_i') + \sum_{i=1}^m \text{size}(t_i)\right) + \text{size}(p) \cdot 2n_c\lceil \log_2 N \rceil
\leq 2\text{size}(p)^2
\end{align*}
where $2n_c\lceil \log_2 N \rceil \leq \text{size}(p)$ because each \texttt{if}-condition corresponds to an occurrence of itself and an \texttt{if-then-else} operator: Both of them contributes to $\text{size}(p)$ by at least $\lceil \log_2 N \rceil$.
\end{proof}

\begin{lemma} [Lemma \ref{lemma:unifier-gen}]$\mathcal C$ is an $(\alpha, \beta)$-Occam solver on $\text{DNF}(\mathcal F_C) \Rightarrow \mathcal U_{\text{poly}}$ is an $(\alpha', \beta)$-Occam unifier on $\mathcal F_C$ for any $\alpha' > 4\alpha$, where $\text{DNF}(\mathcal F_C)$ is defined as $\{(\text{DNF}(\mathbb P_c), \mathbb I')\ |\ ((\mathbb P_t, \mathbb P_c), \mathbb I) \in \mathcal F_C, \mathbb I' \subseteq \mathbb I\}$. 
	
	$\mathcal C$ is a deterministic $(\alpha, \beta)$-Occam solver on $\text{DNF}(\mathcal F_C) \Rightarrow \mathcal U_{\text{poly}}$ is a $(4\alpha, \beta)$-Occam unifier on $\mathcal F_C$.
\end{lemma}
\begin{proof} Let $P= \{t_1, \dots, t_m\}$ be the term set and $p^* \in (P, \mathbb P_c)$ be the target program. Let $c_1^*, \dots, c_{m-1}^*$ be the \texttt{if}-conditions for $t_i$ in the same way as the proof of Lemma \ref{lemma:decision-tree-size}. Specially, if term $t_i$ is not used in $p^*$, $t_i$ is defined as $\texttt{false}$. Let $s = \max(\text{tsize}(P), \text{size}(p^*))$. By the construction of $c_i^*$, we have (1) $\sum_{i=1}^{m-1} \text{size}(c_i^*) \leq O(s^2)$, and (2) $c_1^*, \dots, c_{m-1}^*$ never overlap, i.e., $\forall 1 < i < j < m, \forall I \in \mathbb I, \neg (\sem{c_i^*}(I) \wedge \sem{c_j^*}(I))$.
	
Let $T_i$ be the set of examples that are satisfied by term $t_i$ and are not satisfied by any term in $t_{i+1}, \dots, t_{m}$. By Algorithm \ref{alg:uni-over}, the set of positive examples provided to synthesize $c_i^*$ must be a subset of $T_i$. As $c_i^*$ may not satisfy all positive examples in $T_i$, we construct conditions $c'_1, \dots, c'_{m-1}$ instead:
$$
c'_i \coloneqq c_i^* \texttt{ or } \left(\text{or}_{j=1}^{i-1}\ \left(\texttt{or}_{k=1}^{n_j}\ \left(c_{j,k}^* \texttt{ and } \left(t_i = t_j\right)\right)\right)\right)
$$
where $c_{i,1}^*, \dots, c_{i,n_i}^*$ are the clauses used in condition $c_i^*$, i.e., $c_i^* = \texttt{or}_{j=1}^{n_i}\ c_{i,j}^*$.
Clearly, $c_i'$ is still a DNF formula, and $c_i'$ satisfies all examples in $T_i$. Therefore, $c_i'$ is always a target condition for the PBE task corresponding to $c_i$. As in $c'_1, \dots, c'_{m-1}$, each clause $c_{i,j}^*$ in condition $c_i^*$ occurs at most $m-1$ times, each time with a comparison $t_i = t_k$. Therefore, $\sum_{i=1}^{m-1} \text{size}(c_i') = O(s^4)$.

Let $c_1, \dots, c_{m-1}$ be the condition synthesized by $\mathcal C$. Let $\epsilon$ be any constant in $\left(0, \frac{1}{2}\right)$. For simplicity, we abbreviate $\text{size}(c_i')$ as $s_i$. When $\mathcal C$ is an $(\alpha, \beta)$-Occam solver, there exist constants $c', \gamma'$ such that:
$$
\forall i \in [1, m - 1], \Pr\left[\text{size}(c_i) > c's_i^{\alpha}|T|^{\beta}\ln^{\gamma'}\left(\frac{m-1}{\epsilon}\right)\right] \leq \frac{\epsilon}{m-1}
$$

Therefore, with a probability of at least $1-\epsilon$:
\begin{align*}
\sum_{i=1}^{m-1} \text{size}(c_i) &\leq c'|T|^{\beta}\ln^{\gamma'}\left(\frac{m-1}{\epsilon}\right) \sum_{i=1}^{m-1} s_i^{\alpha} \leq c'|T|^{\beta}\left(\ln(m - 1) + \ln \left(\frac{1}{\epsilon}\right)\right)^{\gamma'} \left(\sum_{i=1}^{m-1} s_i\right)^{\alpha} \\
& \leq c''|T|^{\beta}s^{4\alpha}\ln(m-1)^{\gamma'}\ln\left(\frac{1}{\epsilon}\right)^{\gamma'} \leq c s^{\alpha'}|T|^{\beta}\ln\left(\frac{1}{\epsilon}\right)^{\gamma}
\end{align*}
where $c''$ and $c$ are large enough constants, $\gamma = \gamma'$ and $\alpha'$ is a constant larger than $4 \alpha$. Let $p$ be the program synthesized by $\mathcal U_{\text{poly}}$. Then $\text{size}(p) \leq \sum_{i=1}^{m-1} \text{size}(c_i) + s$. Therefore, when $\mathcal C_{\text{CL}}$ is an $(\alpha,\beta)$-Occam solver, $\mathcal U_{\text{poly}}$ must be a $(\alpha', \beta)$ unifier for any $\alpha' > 4\alpha$.

When $\mathcal C_{\text{CL}}$ is a deterministic $(\alpha, \beta)$-Occam solver, there exists constant $c'$ such that:
$$
\forall i \in [1, m - 1], \text{size}(c_i) \leq c's_i^{\alpha}|T|^{\beta}
$$

At this time, we have the following bound on the total size of conditions $c_1, \dots, c_{m-1}$.
$$
\sum_{i=1}^{m-1} \text{size}(c_i) \leq c's_i^{\alpha}|T|^{\beta} \leq c'|T|^{\beta}\sum_{i=1}^{m-1}s_i^{\alpha} \leq c'|T|^{\beta}\left(\sum_{i=1}^{m-1}s_i\right)^{\alpha} \leq cs^{4\alpha}|T|^{\beta}
$$
where $c$ is a large enough constant. Therefore, at this time, $\mathcal U_{\text{poly}}$ is a $(4\alpha, \beta)$-Occam unifier.
\end{proof}

\begin{lemma} [Lemma \ref{lemma:cl-size}] Given condition space $\mathbb P_c$ and PBE task $T$, let $c^*$ be the smallest valid clause and $c$ be the clause found by $\mathcal C_{\text{CL}}$. Then $\text{size}(p_c(c)) < 2\text{size}(p_c(c^*)) (\ln |T| + 1)$.
\end{lemma}
\begin{proof} Let $L$ be the set of available literals, and $L^* = \{l_1, \dots, l_n\}$ be the set of literals satisfying all positive examples in $T$. Then, $c$ and $c^*$ must be subsets of $L^*$.
	
	By Algorithm \ref{alg:uni-over}, while applying the greedily algorithm for set covering, we set the cost of a literal as its size. Therefore, by the approximation ratio of this algorithm, we have:
	$$
	\sum_{l \in c} \text{size}(l) < \left(\ln \left|\mathbb I_N(T)\right| + 1\right) \sum_{l \in c^*} \text{size}(l) \leq \left(\ln \left|T\right| + 1\right) \sum_{l \in c^*} \text{size}(l)
	$$
	
By this inequality, we have:
\begin{align*}
\text{size}(p_c(c)) &= \left(|c| - 1\right) \lceil \log_2 N \rceil + \sum_{l \in c} \text{size}(l)< 2 \sum_{l \in c} \text{size}(l) \\&< 2 \left(\ln \left|T\right| + 1\right) \sum_{l \in c^*} \text{size}(l) < 2\text{size}(p_c(c^*)) (\ln |T| + 1)
\end{align*}
\end{proof}

\begin{corollary}[Corollary \ref{corollary:cl-size}]
	For any $0 < \beta < 1$, $\mathcal C_{\text{CL}}$ is an $(1, \beta)$-Occam solver on all possible clause domains.
\end{corollary}
\begin{proof} For any conditional space $\mathbb P_c$ and PBE task $T$, let $c^*$ be the target clause and $c$ be the clause synthesized by $\mathcal C_{\text{CL}}$. By Lemma \ref{lemma:cl-size}, we have:
	$$
	\text{size}(p_c(c)) < 2(\ln|T| + 1)\text{size}(p_c(c^*)) < c\text{size}(p_c(c^*))|T|^{\beta}
	$$
where $c$ is a large enough constant and $\beta$ is a constant in $(0, 1)$. Therefore, for any $0 < \beta < 1$, $\mathcal C_{\text{CL}}$ is an $(1, \beta)$-Occam solver.
\end{proof}

\begin{lemma}[Lemma \ref{lemma:or-sample}] Let $T$ be a PBE task and $d$ be a DNF formula satisfying all examples in $T$:
	\begin{itemize}
		\item All clauses in $d$ must be \texttt{false} on all negative examples in $T$, i.e., $\forall c \in d, \mathbb I_N(T) \subseteq N(\mathbb I(T), c)$.
		\item There exists a clause in $d$ that is \texttt{true} on at least $|d|^{-1}$ portion of positive examples in $T$, i.e., $\exists c \in d, |P(\mathbb I(T), c)| \geq |d|^{-1}|\mathbb I_P(T)|$.
	\end{itemize}
\end{lemma}
\begin{proof} By the semantics of operator \texttt{or}, the first condition is obtained directly. Let $c_1, \dots, c_m$ be the clauses in $d$, and $w_1, \dots, w_m$ be the number of positive examples covered by each clause, i.e., $w_i \coloneqq |P(\mathbb I(T), c_i)|$. By the semantics of opeartor \texttt{or}, we know each positive example must be covered by at least one clause. Therefore:
	$$
	w_1 + \dots + w_m \geq |\mathbb I_P(T)| \implies \max w_i \geq m^{-1}|\mathbb I_P(T)|
	$$
In this way, the second condition is obtained.
\end{proof}

\begin{lemma} [Lemma \ref{lemma:condition-occam}] For any $0 < \beta < 1$, $\mathcal C$ is a $(2, \beta)$-Occam solver on $\text{DNF}(\mathcal F_C^A)$. 
\end{lemma}
\SetKwFunction{G}{Get}
\begin{proof} Let $s^*$ be the value of variable $s$ (Line 10 in Algorithm \ref{alg:uni-dnf-main}) when $\mathcal U_{\text{poly}}$ terminates. Let $d^*$ be the target DNF formula. When $s \geq 2\text{size}(p_d(d^*))$, $s' = s$ and the initial value of $k$ is also $s$, we make the following claim:
\begin{itemize}
	\item \textbf{Claim:} In invocation \Search{$\variable{literals}, T, k, s$}, let $d'$ be the set of clauses $c^*$ in $d^*$ satisfying $P(\mathbb I_P(T), c^*) \neq \emptyset$. If $k \geq |d'|$, there must be a clause $c \in \G{$\variable{literals}, T, k$}$ such that $\exists c^* \in d', P(\mathbb I_P(T), c^*) \subseteq P(\mathbb I_P(T), c) \wedge \text{size}(p_c(c)) \leq 2s \ln |T|$. 
	\item \textbf{Proof:} Let $c^*$ be the clause in $d'$ that covers the most positive examples. Clearly, $|P(\mathbb I_P(T), c^*)| \geq |d'|^{-1}|\mathbb I_p(T)|$. As $k \geq d'$, the largest clause in $[c^*]_{\mathbb I_P(T)}$ must be in $R(\mathbb I_P(T), k, \variable{literals})$, and thus this clause is found by $\G{$\variable{literals}, T, k$}$.  

	Let $c$ be the clause simplified from the largest clause in $[c^*]_{\mathbb I_P(T)}$. By Lemma \ref{lemma:cl-size}, $\text{size}(p_c(c)) < 2\text{size}(p_c(c^*))(\ln |T| + 1) \leq 2s \ln |T|$. 
\end{itemize}

By this claim, we obtain that in an iteration where $s > 2\text{size}(p_d(d^*))$, $\Search{}$ can always find a set of clauses. So, there are at most $O(\text{size}(p_d(d^*)))$ clauses in the DNF formula synthesized by $\mathcal C$, and the size of each clause is at most $O(\text{size}(p_d(d^*))\ln |T|)$. Therefore, the size of the synthesized DNF formula is $O((\text{size}(p_d(d^*)))^2\ln |T|)$, which directly implies that $\mathcal C$ is a $(2, \beta)$-Occam solver for any $0 < \beta < 1$.

\end{proof}

\begin{theorem} [Theorem \ref{theorem:unifier-occam}] For any $0 < \beta < 1$, $\mathcal U_{\text{poly}}$ is a $(8, \beta)$-Occam unifier on $\mathcal F_C^A$.
\end{theorem}
\begin{proof} Directly by Lemma \ref{lemma:sample} and Lemma \ref{lemma:condition-occam}.
\end{proof}

\begin{theorem} [Theorem \ref{theorem:final}] $\mathcal S_t$ is an $(\alpha, \beta)$-Occam solver on $T(\mathcal F_C)$ with $\beta < \frac{1}{8} \implies$\mainname is an $(8(\alpha' + 1), 8\beta')$-Occam solver on $\mathcal F_C$ for any $\alpha' > \alpha, \beta < \beta' < \frac{1}{8}$.
\end{theorem}
\begin{proof} Directly by Theorem \ref{theorem:gen-stun}, Theorem \ref{theorem:term-occam} and Theorem \ref{theorem:unifier-occam}.
\end{proof}

\end{document}